\documentclass[review, nopreprintline, authoryear, 11pt]{elsarticle}
\usepackage[a4paper, left=2cm, right=2cm, top=2.2cm, bottom=2.2cm]{geometry}

\usepackage[english]{babel}
\usepackage[utf8]{inputenc}
\usepackage[T1]{fontenc}
\usepackage{color}
\usepackage{lmodern, microtype}

\usepackage{amsmath}
\usepackage{amsthm}
\usepackage{amstext}
\usepackage{amssymb}
\usepackage{latexsym}
\usepackage{dsfont}
\usepackage{mathrsfs}

\usepackage{caption}
\usepackage{subcaption}

\usepackage{setspace}
\usepackage[normalem]{ulem}
\usepackage{bbold}
\usepackage{natbib}
\setcitestyle{authoryear,open={(},close={)}} %Citation-related commands

\DeclareMathAlphabet{\mathbb}{U}{msb}{m}{n}

\usepackage{hyperref}
\hypersetup{
    colorlinks=true,
    unicode=true
}

\usepackage{graphicx}
\usepackage{float}
\usepackage{caption, booktabs}
\usepackage{graphics}
\graphicspath{{./plots_new/}}
\usepackage{subcaption}
\usepackage{lipsum}
\usepackage{amsmath,amssymb,amsthm,amsfonts}
\usepackage[table, dvipsnames]{xcolor}
\usepackage{mathrsfs}
\usepackage{soul}
\usepackage{enumitem}

\usepackage{threeparttable}
\newtheorem{theorem}{Theorem}
\newtheorem{prop}{Proposition}

\newtheorem{example}{Example}
\newtheorem{remarks}{Remarks}
\newtheorem{condition}[theorem]{Condition}

\theoremstyle{definition}

\newcolumntype{H}{>{\setbox0=\hbox\bgroup}c<{\egroup}@{}}

\usepackage{algorithm2e}
\RestyleAlgo{ruled}
\usepackage{algpseudocode}
\usepackage{dsfont}

\newcommand{\SP}[1]{\color{black}{{#1}} \color{black}}

\newcommand{\new}[1]{\textcolor{black}{#1}}

\def\rmd{\,{\rm d}}  % for differential in integrals

\usepackage{tikz}
\usepackage{lineno}
\modulolinenumbers[0]

\begin{document}

% Comment out to disable line numbers
%\linenumbers

\onehalfspacing

\begin{frontmatter}

    \journal{
    }

    \title{Copula-Based Time Series for Non-Gaussian and Non-Markovian Stationary Processes}


    \author[1]{Sven Pappert\corref{cor1}}
    \ead{sven.pappert@ruhr-uni-bochum.de}
    \cortext[cor1]{Corresponding author}
\address[1]{Faculty of Mathematics, Ruhr University Bochum, Germany}

	\author[2]{Harry Joe}
    \ead{Harry.Joe@ubc.ca}
\address[2]{Department of Statistics, University of British Columbia, Canada}

    \begin{abstract}
   In the copula-based approach to univariate time series modeling, the finite-dimensional temporal dependence of a stationary time series is captured by a copula. Recent studies investigate how copula-based time series models can be generalized to have long-term autoregressive effects. We study a generalization that comes from a Markov sequence of order $p$ and a $q$-dependent sequence. We derive the relation of the model to Gaussian-ARMA models and to the Gaussian-GARCH$(1,1)$ model. We investigate distributional properties of the process and discuss the maximum likelihood estimation (MLE). Additionally we analyze the copula moving aggregate process of order one, or MAG$(1)$, as it is a basic building block. Last, we assess the model in a probabilistic forecasting study on a data series for German wind power production.
    \end{abstract}
    \begin{keyword}
        Copula-based time series, Gaussian ARMA, Identifiability, Markov dependence, Probabilistic Forecasting, Tail dependence
    \end{keyword}
\end{frontmatter}

\section{Introduction}
\label{Sec:Introduction}

% General introduction to Copula-based TS
Copula-based time series models have been used for non-Gaussian univariate time series. The joint distribution of $(p+1)$ consecutive observations is decomposed into the copula and the unconditional stationary distribution. If the underlying process is a Markov process of order $p$, the copula captures the serial dependence of the process. The model allows for non-linear serial dependence and can model clustering of extreme events properly by assigning a copula with suitable tail dependence. Additionally, the modeling of the unconditional stationary distribution is completely decoupled from the modeling of serial dependence, resulting in more flexibility. Copula-based time series models (or multivariate version thereof) are subject to extensive research and are used for applications.
% Some examples from related literature
To this end, the model is succesfully employed for forecasting of volatile time series by e.g. \cite{loaiza2018time} and \cite{pappert2023forecasting}. \cite{chen2006estimation, nagler2022stationary} derive asymptotic properties of (semi-parametric) estimation in the univariate and the multivariate case respectively. 
%\cite{beare2010copulas, ibragimov2005copula} investigate the forms of temporal dependencies,
\cite{simard2015forecasting, brechmann2015copar, nagler2022stationary, beare2015vine} consider multivariate extensions. Last, there is an emerging research brand on generalizing the model to have long-term serial dependence by \cite{mcneil2022time, pappert2024moving, bladt2025semiparametric}.
%Contextualization
This direction of research is also the context in which this work lies is. The purpose of this work is to explore a generalization that allows the copula-based time series model to have long-term autoregressive or non-Markov effects. As it turns out, this generalization is equivalent to finding a copula-based generalization of the autoregressive moving average process (ARMA). We are interested in this generalization for the following reason.
When the process that is to be modeled is not a Markov process of order $p$, modeling the distribution of $(p+1)$ consecutive observations is not sufficient to capture the serial dependence. In fact, if the true data-generating process has asymptotically decreasing autocorrelation, e.g. ARMA$(1,1)$, no finite dimensional distribution is sufficient to capture the serial dependence. Additionally, we may be able to model processes with large $p$, that is a large autoregressive order, in a more parsimonious way. Now we discuss the existing approaches.\\
% copula-arma generalizations
In the related literature three attempts for generalizing copula-based time series models are available. First, \cite{joe2014dependence}, in Sect. 3.14, proposes an approach is inspired by the Poisson-ARMA model by \cite{mckenzie1988some}. The basic idea is to model a time series by combining a AR$(p)$ copula-based time series model with a $q$-dependent copula-based time series model. The AR time series is a latent process in this model and serves as an innovation for the $q$-dependent time series. Second, \cite{mcneil2022time} and \cite{bladt2025semiparametric} propose to consider an AR$(\infty)$ copula-based time series model and to reparametrize in terms of linear ARMA$(p,q)$ parameters, where $p,q < \infty$. The authors conjecture the well-behavedness of the emerging time series based on the well-behavedness of the underlying linear ARMA model.
Third, \cite{pappert2024moving} proposes the moving aggregate modified autoregressive copula-based time series model, 
which similarly to \cite{joe2014dependence} (Sect.~3.14) introduces a $q$-dependent or moving aggregate component.
The model is also combined with an AR copula-based time series model. 
In contrast to the model by \cite{joe2014dependence} \SP{(Sect. 3.14),} the 
%$q$-dependent 
moving aggregate process is the latent process and is used as innovation process for the AR process. A major problem of this model is that, in its pure form, it does not have a uniform distribution as unconditional stationary distribution. This problem essentially prevents proper marginal modeling. The solution proposed by \cite{pappert2024moving} is to introduce an auxiliary transformation that ensures uniformity. However, it is cumbersome to properly employ this transformation.
In summary, all of these models enable copula-based time series models to have infinite memory. Additionally, all these models can be understood as non-linear and probabilistic generalizations of the classic linear ARMA model.
%\\

% The contribution
One purpose of this paper is to examine and advance the theory for the model proposed in Sect. 3.14 by \cite{joe2014dependence} and related processes.
%In its original publication, 
The model was introduced, with derivations of the joint distribution of the time series and the innovations, and the likelihood. 
However, no dependence properties were obtained, and the link to Gaussian-ARMA was claimed without details.
The model is given in terms of the following time series updating equations,
\begin{align}
U_t = h(\varepsilon_t,\hdots,\varepsilon_{t-q+1}, W_{t-q}), \quad
W_t = g(\varepsilon_t,W_{t-1},\hdots,W_{t-p}),
\quad
\varepsilon_t \stackrel{iid}{\sim} \, U(0,1).
\label{Eq:Model(p,q)_updating_Eq}
\end{align}
\SP{where $h$ is a conditional quantile function corresponding to a copula $K$ and $g$ is also a conditional quantile functions associated to a copula $C$.}
As an example for a \SP{conditional} quantile function corresponding to a copula, we give the quantile function of a bivariate Gaussian copula. But first, we give the conditional cumulative distribution function (CDF). The conditional CDF corresponding to a bivariate Gaussian copula $A$ with parameter $\alpha$ is given as (cf.~\cite{smith2010modeling})
\begin{align*}
A_{2|1}(u|v) = \Phi\left( \frac{\Phi^{-1}(u) - \alpha \Phi^{-1}(v)}{\sqrt{1 - \alpha^2}} \right), \quad u,v \in (0,1).
\end{align*}
The conditional quantile function can be derived by inverting as
\begin{align*}
A_{2|1}^{-1}(s|v) = \Phi\left( \sqrt{1-\alpha^2} \Phi^{-1}(s) + \alpha \Phi^{-1}(v) \right), \quad s,v \in (0,1)
\end{align*}
This \SP{expression} will play an important role for the Gaussian-MAG$(1)$ process and the relation of the model under investigation to Gaussian-ARMA models, discussed later.
%
% Properties
%
Since $(W_{t-q},\varepsilon_{t-q+1},\hdots,\varepsilon_t)$ in Eq.~\ref{Eq:Model(p,q)_updating_Eq} are independent and standard uniformly distributed, it follows that $U_t \sim U(0,1)$. Consequently, any desired unconditional stationary distribution can be obtained by quantile-transforming $\{U_t\}$. This property is an advantage to the model by \cite{pappert2024moving}, which needs an additional transformation to ensure uniformity.
% Precise research questions:
These features motivate us to further investigate the process defined by Eq.~\ref{Eq:Model(p,q)_updating_Eq}. We make the following contributions:

\begin{itemize}
\item[a)] First, we examine in detail the $1$-dependent copula-based time series model, i.e., $\{V_t\}_{t\in\mathbb{Z}}$ given as $V_t = h(\varepsilon_{t}, \varepsilon_{t-1})$, \SP{with $h$ being the conditional quantile function associated to a copula $K_{12}$.} This process is the building block of the model under investigation and also the MAGMAR-copula model by \cite{pappert2024moving}. We show that the time series inherits some dependence properties from the MAG$(1)$-copula, such as positive quadrant dependence (PQD) and stochastic increasingness (SI'ness). However, it seems that the tail dependence coefficients of $(V_t, V_{t-1})$ is restricted. A known result on the form of tail dependence matrices by \cite{embrechts2016bernoulli} gives an absolute upper bound as $\frac{1}{2}$. Using numerical approximations and simulations we observe that most MAG$(1)$-copula time series only exhibit a little to potentially no tail dependence in the joint distribution of consecutive observations.
\item[b)] We show that a (subset of) Gaussian-ARMA processes is recovered if the AR and MAG copulas are chosen as Gaussian and the time series is quantile transformed by the Gaussian quantile function. However, there are subtleties. For general model orders $p,q > 1$, a subset of a Gaussian-ARMA$(p,q+p-1)$ is recovered. Due to the structure of the process additional terms arise that cause the distorted MA-order.
\item[c)] We derive the copulas necessary to recover the ARCH and GARCH process. The copulas of a GARCH process are also investigated by \cite{dias2024garch}. The difference to our approach is that their approach captures the serial dependence employing the infinite partial dependence model by \cite{mcneil2022time}. We capture the dependence by an appropriate choice of the AR and MAG copula \SP{for the model in Eq.~\ref{Eq:Model(p,q)_updating_Eq}.} For the case where $q=0$, i.e., the ARCH-model, our results coincide with the one from \cite{dias2024garch} (albeit, we give a different derivation).
\item[d)] We derive the copula of consecutive observations from the model in Eq.~\ref{Eq:Model(p,q)_updating_Eq} with model orders $(p,q) \in \{(1,1), (1,2),(2,1), (2,2)\}$.
\item[e)] With regards to estimation, we first discuss the (near) non-identifiability of the classic linear MA$(1)$ and our copula-based MAG$(1)$. We show that, similar to the classic MA$(1)$, the \SP{Gaussian-copula} version also has two representations. The two representation can be obtained by permuting the order of the innovations in the updating equation. We explore potential identifiability problems.
\item[f)] We discuss conditions for the consistency of the maximum-likelihood estimators of the process from Eq.~\ref{Eq:Model(p,q)_updating_Eq}, following the approach by \cite{pappert2024moving}. The crucial step is to find conditions to establish the ergodicity of the score-time series based on the theory of stochastic recurrence equations, cf.~\cite{douc2014nonlinear}. Then general results for M-estimators for dependent data can be used, cf.~\cite{wooldridge1994estimation}. Currently the conditions are on a high-level and a future research direction could be to translate these high-level conditions to conditions on the copula parameters.
\item[g)] Last, we model and forecast a time series using the model from Eq.~\ref{Eq:Model(p,q)_updating_Eq}. The time series of interest is German wind-power production. The \texttt{R} code for estimation, simulation and forecasting is available upon request. The implementation is based on the \texttt{rvinecopulib} package \new{(see \cite{nagler2022rvinecopulib})}.
\end{itemize}

In remainder of this manuscript, we \SP{formally} introduce the model in Sect.~\ref{Sec:Model}. We proceed to derive the relation to other models, theoretical properties and discuss the iterative estimation. In Sect.~\ref{Sec:Nums_Sims} we use simulations and numerical approximations to explore properties of the MAG$(1)$ and the model in Eq.~\ref{Eq:Model(p,q)_updating_Eq}. In Sect.~\ref{Sec:Forecasting}, we use the model for probabilistic forecasting and in Sect.~\ref{Sec:Discussion} we discuss our findings.

\section{Copula-ARMA generalization}
\label{Sec:Model}

In Eq.~\ref{Eq:Model(p,q)_updating_Eq} the mappings $g$ and $h$ are conditional quantile functions corresponding to the AR- and MAG copulas, $C$ and $K$, respectively. The copula $C$ is a D-vine copula where all copulas in a tree are set to be the same copula. This condition is placed in order to assure stationarity. This condition may also be called \textit{translational invariance}, cf.~\cite{nagler2022stationary}. We may call vine-copulas of this form \textit{stationary D-vine} or \textit{Toeplitz vine}. Due to the D-vine parametrization, the parameters governing the influences from different lags are algebraically independent. For illustration, we show the contour plot of such vine-copula in Fig.~\ref{Fig:contour_sd_vine}. 
The copula $K$ is set to be a D-vine copula with independent margins. This means that $K$ is the copula of $(U_t, \varepsilon_{t-1},\hdots,\varepsilon_{t-q+1}, W_{t-q})$ with $( \varepsilon_{t-1},\hdots,\varepsilon_{t-q+1}, W_{t-q})$ being independent, \SP{compare Sect. 3.14 by \cite{joe2014dependence}.} $K$ consists of $q$ bivariate copulas, which govern the dependence between $(U_t,\varepsilon_{t-1})$, $(U_t,\varepsilon_{t-2}|\varepsilon_{t-1})$ up to $(U_t, W_{t-q}|\varepsilon_{t-1},\hdots,\varepsilon_{t-q+1})$ with algebraically independent parameters.
Such vine copulas may be called \textit{MAG D-vine} or $q$\textit{-dependent D-vine}. For illustration, we show the contour of such a vine-copula in Fig.~\ref{Fig:contour_MA_vine}. For the special case where $K = C^{\perp}$, the time series simplifies to $U_t = \varepsilon_t$, i.e., a standard uniform independent and identically distributed \SP{($i.i.d.$)} process. For $p =q = 1$ and $K = C^+$, we see that $U_t = W_{t-1} = C_{2|1}^{-1}(\varepsilon_{t-1}|W_{t-2}) = C_{2|1}^{-1}(\varepsilon_{t-1}|U_{t-1})$. This is the copula AR process with shifted innovations. The shift in the innovations does not affect the properties of the time series at all.
\begin{figure}[t]
\centering
\includegraphics[scale=0.65]{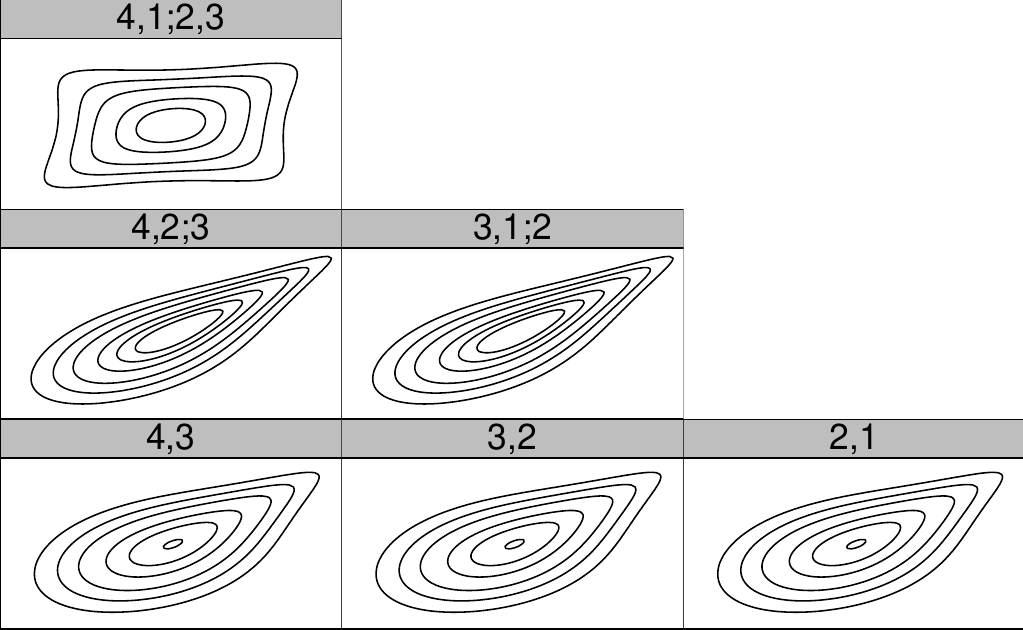}
\caption{Contour plots of a stationary D-vine corresponding to an autoregressive copula-based time series model of order $p = 3$.}
\label{Fig:contour_sd_vine}
\end{figure}
\begin{figure}[t]
\centering
\includegraphics[scale=0.65]{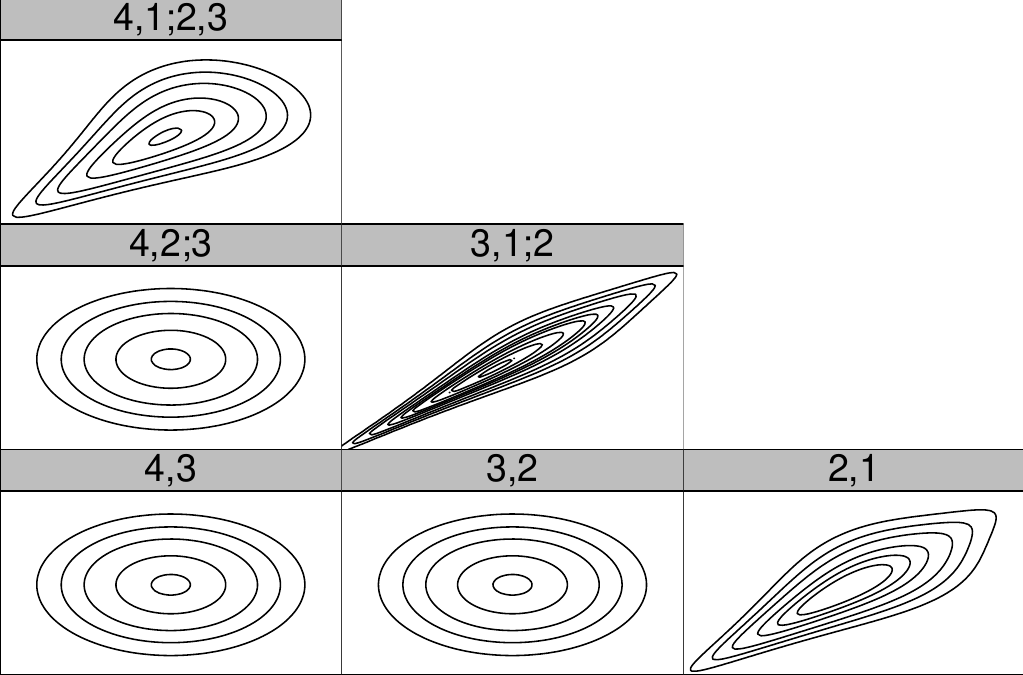}
\caption{Contour plots of a MAG D-vine corresponding to a $q$-dependent copula-based time series model of order $q = 3$. Only dependencies involving variable $1$ are different from independence.}
\label{Fig:contour_MA_vine}
\end{figure}

\noindent
In this section, we first derive the relation between the model from Eq.~\ref{Eq:Model(p,q)_updating_Eq} and the linear Gaussian-ARMA and the Gaussian-GARCH$(1,1)$ model.
Then we proceed to investigate distributional properties such as stationarity, ergodicity, joint and conditional distributions, tail dependencies and measures of association. Last, we discuss the maximum likelihood estimation.

\subsection{Relation to other models}

In this section, we derive the relation between the model in Eq.~\ref{Eq:Model(p,q)_updating_Eq} and popular time series models. In particular, we first show that for the AR and the MAG copula being the Gaussian copula, the transformed time series $\{\Phi^{-1}(U_t)\}_{t\in\mathbb{Z}}$ follows a Gaussian ARMA process. The relation between the parameters of the copula and the parameters of the emerging Gaussian ARMA is not trivial. Second, we show that we can recover the GARCH model with a suitable choice of copulas. We formalize the relation to the Gaussian ARMA process in the following proposition. We note that part a) is not new, but we include still for the sake of completeness.

\begin{prop}[Recovering Gaussian ARMA]
\label{prop:GaussianARMA}
Let $\{U_t\}_{t\in\mathbb{Z}}$ be defined through Eq.~\ref{Eq:Model(p,q)_updating_Eq} with latent process $W_t \sim U(0,1)$ and innovations $\varepsilon_t \stackrel{iid}{\sim} U(0,1)$. Let the copula \SP{$C$,} corresponding to the AR-component be given as a $(p+1)$-dimensional stationary Gaussian D-vine copula, which has a bivariate Gaussian copula with parameter $\alpha_j$ for lag $j \in \{1,\hdots,p\}$. Let the copula \SP{$K$,} corresponding to the MA-component be given as a $(q+1)$-dimensional Gaussian MAG D-vine, which has a bivariate copula with parameter $\beta_j$ for lag $j \in \{1,\hdots,q\}$.
Furthermore, let the time series $\{Y_t\}_{t\in\mathbb{Z}}$ be given as $Y_t = \Phi^{-1}(U_t)$ and the time series $\{X_t\}_{t\in\mathbb{Z}}$ as $X_t = \Phi^{-1}(W_t)$. Then the following holds.
\begin{itemize}
\item[a)] The time series $\{X_t\}_{t\in\mathbb{Z}}$ is a Gaussian AR$(p)$ process following the updating equation
\begin{align*}
X_t = \sum_{j = 1}^p \phi^{(p)}_j X_{t-j} + \xi_t,
\end{align*}
where $\xi_t = \sigma_{\alpha} \Phi^{-1}(\varepsilon_t) \stackrel{iid}{\sim} N(0, \prod_j(1-\alpha_j^2))$ with $ \sigma_{\alpha} = \prod_j(1-\alpha_j^2)^{1/2}$ being the standard deviation of $\xi_t$ so that $\text{Var}(X_t) = 1$ for all $t$. The parameters $\phi^{(p)}_j$ are defined recursively as,
\begin{align}
    \phi^{(p)}_j =
    \begin{cases}
        \phi_j^{(p-1)} - \alpha_p \phi_{p-j}^{(p-1)}, & j \in \{1,\hdots,p-1\},
        \\
        \alpha_p, & j = p.
    \end{cases}
    \label{Eq:Gaussian_copula_AR(p)_parameters}
\end{align}
Compare Eq. 17 and 18 in \cite{mcneil2022time}, which also provide this result. We still include it for the sake of completeness.
\item[b)] If $p=0$, the process $\{Y_t\}_{t\in\mathbb{Z}} = \{\Phi^{-1}(U_t)\}_{t\in\mathbb{Z}}$ is a Gaussian MA$(q)$-process and follows the updating equation
\begin{align*}
Y_t = \eta_t + \sum_{i = 1}^{q} \theta_i^{(q)} \eta_{t-i},
\end{align*}
where $\eta_t = \sigma_{\beta} \Phi^{-1}(\varepsilon_t) \stackrel{iid}{\sim} N(0, \prod_{j = 1}^q (1 - \beta_j^2))$ with $\sigma_{\beta} = \prod_{j = 1}^q (1 - \beta_j^2)^{1/2}$, which is the standard deviation of $\eta_t$ such that $\text{Var}(Y_t) = 1$ for all $t$. The parameters $\theta_i^{(q)}$ are given as
\begin{align}
\theta_i^{(q)} = \sigma_{\beta}^{-1} \left[ \prod_{j = 0}^{i-1} (1 - \beta_j^2)^{1/2}\right] \beta_i,
\label{Eq:Gaussian_copula_MAG(p)_parameters}
\end{align}
where we set $\beta_0 = 0$.
\item[c)] For arbitrary $p,q \in \mathbb{N}$, the time series $\{Y_t\}_{t\in\mathbb{Z}}$ follows the updating equation,
\begin{align}
Y_t = \eta_t + \sum_{j = 1}^p \phi_j^{(p)} Y_{t-j} + \sum_{j = 1}^{p+q-1} \Psi_{j}^{(p,q)} \eta_{t-j},
\label{Eq:Gaussian_copula_ARMA(p)_parameters}
\end{align}
where the parameters $\phi_j^{(p)}, \ i \in \{1,\hdots,p\}$ are given as in Eq.~\ref{Eq:Gaussian_copula_AR(p)_parameters} and the parameters $\Psi_{j}^{(p,q)}, \ j \in \{1,\hdots,p+q-1\}$ are given as
\begin{align*}
\Psi^{(p,q)}_j =& \tilde{\theta}_j^{(q)} - \tilde{\phi}^{(p)}_j - \zeta_j^{(p,q)},
\\
\tilde{\theta}_j^{(q)} =& \theta_j^{(q)} \times \mathds{1}\{j \leq q \land q > 1\} + \theta_q^{(q)} \sigma_{\alpha} \times \mathds{1}\{j = q\},
\\
\tilde{\phi}_j^{(p)} =& \phi_j^{(p)} \times \mathds{1}\{j \leq p\},
\\
\zeta_j^{(p,q)} =& \left(\sum_{(r,s) \in S_j^{(p,q)}} \phi^{(p)}_r \theta_s^{(q)} \right) \times \mathds{1}\{q > 1\},
\\
S_j^{(p,q)} =& \{(r,s) \in \mathbb{N}^2: r \leq p, s \leq q, r+s = j\},
\end{align*}
with $\theta_j^{(q)}$ as in Eq.~\ref{Eq:Gaussian_copula_MAG(p)_parameters}.
\end{itemize}
\end{prop}

\begin{proof}
The proof is given in~\ref{App:Proof:GaussianARMA}.
\end{proof}

Now we give some examples where we calculate the coefficients from Eqs.~\ref{Eq:Gaussian_copula_AR(p)_parameters}-\ref{Eq:Gaussian_copula_ARMA(p)_parameters} explicitly and do sanity checks by showing that $\text{Var}(Y_t) = 1$ for the respective time series.

\begin{example}[Gaussian-MAG$(1)$]
\label{Example:Gaussian_MA(1)}
Consider $Y_t = \Phi^{-1}(V_t) = \Phi^{-1}(K_{2|1}^{-1}(\varepsilon_t| \varepsilon_{t-1}))$, where $K_{21}$, which is the copula of $(V_t,\varepsilon_{t-1})$, is a Gaussian copula with parameter $\beta_1$ and $\varepsilon_t \stackrel{iid}{\sim} U(0,1)$. Then $q=1$ and $\{Y_t\}_{t\in\mathbb{Z}}$ follows the updating equation $Y_t = \theta_1^{(1)} \eta_{t-1} + \eta_t$. We have $\sqrt{\text{Var}(\eta_t)} = \sigma_{\beta} = (1-\beta_1^2)^{1/2}$ and
\begin{align*}
\theta_1^{(1)} =  (1-\beta_1^2)^{-1/2} [1-\beta^2_0]^{1/2} \beta_1 = \frac{\beta_1}{(1-\beta_1^2)^{1/2}}.
\end{align*}
%As a sanity check, we test whether it holds that $\text{Var}(Y_t) = 1$. In fact
As a check, %we test whether it holds that $\text{Var}(Y_t) = 1$. In fact
\begin{align*}
\text{Var}(Y_t) = \text{Var}(\eta_t) + \bigl(\theta_1^{(1)}\bigr)^2 \text{Var}(\eta_{t-1}) = (1-\beta_1^2) + \frac{\beta_1^2(1-\beta_1^2)}{1 - \beta_1^2} = 1.
\end{align*}
This verifies that the process $V_t$ is in fact $U(0,1)$ and $Y_t \sim N(0,1)$.
\end{example}

\begin{example}[Gaussian-MAG$(2)$]
Consider $Y_t = \Phi^{-1}(V_t) = \Phi^{-1}(K_{3|21}^{-1}(\varepsilon_t|\varepsilon_{t-1}, \varepsilon_{t-2}))$ where $\varepsilon_t \stackrel{iid}{\sim} U(0,1)$ and $K_{321}$ is a $3$-dimensional Gaussian MAG D-vine copula with parameters $\beta_1$ and $\beta_2$. Then the random vector $(V_t,\varepsilon_{t-1})$ has a Gaussian copula with parameter $\beta_1$, the (conditional) random vector $(V_t, \varepsilon_{t-2}|\varepsilon_{t-1})$ has a Gaussian copula with parameter $\beta_2$ and $(\varepsilon_{t-1}, \varepsilon_{t-2})$ are independent. We have $q = 2$ and $\{Y_t\}_{t\in\mathbb{Z}}$ follows the updating equation $Y_t = \theta_{1}^{(2)} \eta_{t-1} + \theta_{2}^{(2)} \eta_{t-2} + \eta_t$ where $\sqrt{\text{Var}(\eta_t)} = \sigma_{\beta} = (1-\beta_1^2)^{1/2}(1-\beta_2^2)^{1/2}$. Furthermore
\begin{align*}
\theta_1^{(2)} =& (1-\beta_1^2)^{-1/2}(1-\beta_2^2)^{-1/2} \beta_1 = \frac{\beta_1}{(1-\beta_1^2)^{1/2}(1-\beta_2^2)^{1/2}}\quad \text{and}
\\
\theta_2^{(2)} =& (1-\beta_1^2)^{-1/2}(1-\beta_2^2)^{-1/2} [1-\beta_1^2] \beta_2 = \frac{\beta_2}{(1-\beta_2^2)^{1/2}}.
\end{align*}
Then $Y_t$ may be represented as
\begin{align*}
Y_t = \eta_t + \frac{\beta_1}{(1-\beta^1)^{1/2}(1-\beta^2)^{1/2}} \eta_{t-1} + \frac{\beta_2}{(1 - \beta_2^2)^{\frac{1}{2}}} \eta_{t-2}
\end{align*}
%We again check whether $\text{Var}(Y_t) = 1$,
As a check,
\begin{align*}
\text{Var}(Y_t) = (1-\beta_1^2)(1-\beta_2^2) + \beta_1^2 + \beta_2^2(1-\beta_1^2) = 1.
\end{align*}
\end{example}

\begin{example}[Gaussian copulas and $p=1$, $q=1$ in Eq.~\ref{Eq:Model(p,q)_updating_Eq}]
\label{Example:Gaussian(1,1)}
Let $\{U_t\}_{t\in\mathbb{Z}}$ follow the model in Eq.~\ref{Eq:Model(p,q)_updating_Eq} with $p = q = 1$. Let the copula corresponding to the AR component be given as bivariate Gaussian with parameter $\alpha$ and let the copula corresponding to the MA component be given as bivariate Gaussian with parameter $\beta$. 
Then $p+q-1 = 1$ and $Y_t = \Phi^{-1}(U_t)$ follows the updating equation
\begin{align*}
Y_t = \phi^{(1)}_1 Y_{t-1} + \Psi^{(1,1)}_1 \eta_{t-1} + \eta_{t},
\end{align*}
where $\sqrt{\text{Var}(\eta_t)} = (1-\beta^2)^{1/2}$ and the AR parameter is given as $\phi^{(1)}_1 = \alpha$. 
The MA parameter is comprised of $\tilde{\theta}_1^{(1)} = \theta_1^{(1)} \sigma_{\alpha} = \beta (1-\alpha^2)^{1/2}/ (1-\beta^2)^{1/2} $ and $\zeta_1^{(1,1)} = 0$ as well as $\tilde{\phi}_{1}^{(1)} = \phi^{(1)}_1 = \alpha$. 
Consequently $\Psi^{(1,1)}_1 = \beta (1-\alpha^2)^{1/2}/ (1-\beta^2)^{1/2}  - \alpha$.
%We again check whether it holds that $\text{Var}(Y_t) = 1$. 
%To this end 
We note that the variance of an ARMA$(1,1)$ process, say $Z_t$, with parameters $(\phi=\phi_1^{(1)},\theta=\Psi_1^{(1,1)},\sigma)$ is given as
\begin{align*}
\text{Var}(Z_t) = \frac{1 + 2\phi\theta + \theta^2}{1 - \phi^2} \sigma^2.
\end{align*}
With a few lines of algebra, it can be checked that $\text{Var}(Y_t)=1$.
%Then it holds that
%\begin{align*}
%    \text{Var}(Y_t) =& \frac{(1 + 2 \phi_1^{(1)} \Psi_1^{(1,1)} + [\Psi_1^{(1,1)}]^2)}{1 - [\phi_1^{(1)}]^2} \sigma^2
%    \\
%    =& \frac{1 + 2\alpha \left[\frac{(1-\alpha^2)^{\frac{1}{2}}\beta}{(1-\beta^2)^{\frac{1}{2}}} - \alpha\right] + \left[\frac{(1-\alpha^2)^{\frac{1}{2}}\beta}{(1-\beta^2)^{\frac{1}{2}}} - \alpha\right]^2}{1 - \alpha^2} \sigma^2
%    \\
%    =& \frac{1 + (-2\alpha^2) + 2\alpha\frac{(1-\alpha^2)^{\frac{1}{2}}\beta}{(1-\beta^2)^{\frac{1}{2}}} + \frac{(1-\alpha^2)\beta^2}{(1-\beta^2)} - 2\alpha\frac{(1-\alpha^2)^{\frac{1}{2}}\beta}{(1-\beta^2)^{\frac{1}{2}}} + \alpha^2}{1-\alpha^2}\sigma^2
%    \\
%    =& \frac{1 -\alpha^2 + \frac{(1-\alpha^2)\beta^2}{(1-\beta^2)}}{1 - \alpha^2}\sigma^2
%    \\
%    =& (1 + \frac{\beta^2}{1 - \beta^2})\sigma^2
%    \\
%    =& \frac{1}{1-\beta^2}\sigma^2 = 1
%\end{align*}
%%\SPP{
%So, in fact, the variance is as expected.\\
The following special cases are of particular interest and can help understand the dynamics of the model: If $\beta = 0$, we obtain $\Psi^{(1,1)}_1 = -\alpha$. The updating equation for $\{Y_t\}_{t\in\mathbb{Z}}$ can then be written as $(1 - \alpha B)Y_t = (1 - \alpha B) \eta_t$, where $B$ is the backshift operator. It follows immediately that $Y_t = \eta_t$. If $\beta \rightarrow 1$, we have $\Psi^{(1,1)}_1 \eta_{t-1} \stackrel{d}{\rightarrow} N(0, (1-\alpha^2))$ and $\eta_t \stackrel{d}{\rightarrow} 0$. The resulting process is a Gaussian AR$(1)$ with parameter $\alpha$. For $\alpha = \beta$ we see that $\Psi^{(1,1)}_1 = 0$ and $\phi^{(1)}_1 = \alpha$. Hence the process also is a Gaussian AR$(1)$ with parameter $\alpha$. These results suggest that with a fixed $0 < \alpha < 1$, the strongest serial dependence occurs for $\beta \in (\alpha, 1)$. 
%This result is confirmed analytically and numerically for the autocorrelation of order one.
This is based on the lag 1 serial correlation derived as
  $$ \rho_1 = (1+\theta\phi)(\phi+\theta)/(1+2\theta\phi+\theta^2)
  =\{ (1-\alpha^2)^{1/2}(1-\beta^2)^{1/2} + \alpha\beta \} \beta. $$
%}
%\texttt{HJ: add serial correlation of lag 1 in terms of $\alpha,\beta$.}
\end{example}
%
%
%
%\SPP{%
We note that the behavior of the process with non-Gaussian copulas is similar with regards to the special cases discussed at the end of Example~\ref{Example:Gaussian(1,1)}. The results can be derived directly from the stochastic representation of the model, i.e., Eq.~\ref{Eq:Model(p,q)_updating_Eq} with $p=q=1$: For $K_{21} = C^{\perp}$ it follows that $U_t = \varepsilon_t$ for all $t$, irrespective of the choice of $C_{21}$. For $K_{21} = C^+$ we see that $U_t = W_{t-1}$ and for $K_{21} = C^{-}$ it holds that $U_t = 1-W_{t-1}$  for all $t$. For $C_{21} = C^{\perp}$, we recover the MAG$(1)$ model. If $K_{12} = C_{12}$, we have $U_t = W_t$ with Markov order $1$ for all $t$.
%}
%
%
%

\begin{example}[Gaussian copulas and $p=2$, $q=1$ in Eq.~\ref{Eq:Model(p,q)_updating_Eq}]
Let $\{U_t\}_{t\in\mathbb{Z}}$ follow the model in Eq.~\ref{Eq:Model(p,q)_updating_Eq} with $p = 2$ and $q = 1$. Let the copula corresponding to the AR component, $C_{321}$, be given as Gaussian with parameters $(\alpha_1, \alpha_2)$ and let the copula corresponding to the MA component be given as bivariate Gaussian with parameter $\beta$. Then we have $p+q-1 = 2$ and $\{Y_t\}_{t\in\mathbb{Z}} = \{\Phi^{-1}(U_t)\}_{t\in\mathbb{Z}}$ has updating equation
\begin{align*}
Y_t = \phi^{(2)}_1 Y_{t-1} + \phi^{(2)}_2 Y_{t-2} + \Psi_1^{(2,1)} \eta_{t-1} + \Psi_2^{(2,1)} \eta_{t-2} + \eta_t,
\end{align*}
where the AR parameters are given as $\phi^{(2)}_1 = \alpha_1 - \alpha_2 \alpha_1 = \alpha_1(1-\alpha_2)$ and $\phi^{(2)}_2 = \alpha_2$. For $\Psi_1^{(2,1)}$ we we see that $\tilde{\phi}_1^{(2)} = \phi_1^{(2)} = \alpha_1(1-\alpha_2)$ and $\tilde{\theta}_1^{(1)} = \theta_1^{(1)} \sigma_{\alpha} = (1-\beta_1^2)^{-1/2}\beta_1 (1-\alpha_1^2)^{1/2}(1-\alpha_2^2)^{1/2}$. Furthermore $\mathds{1}\{q>1\} = 0$ and hence $\zeta_{1/2}^{(2,1)} = 0$. Now for $\Psi_2^{(2,1)}$ we have $\tilde{\phi}^{(2)}_2 = \phi^{(2)}_2 = \alpha_2$ and $\tilde{\theta}^{(1)}_2 = 0$. In total we have
\begin{align*}
Y_t = [\alpha_1(1-\alpha_2)] Y_{t-1} + \alpha_2 Y_{t-2} + \left[ \frac{\beta_1 (1-\alpha_1^2)^{1/2} (1-\alpha_2^2)^{1/2}}{(1-\beta_1^2)^{1/2}} - \alpha_1(1-\alpha_2) \right] \eta_{t-1} + [-\alpha_2]\eta_{t-2} + \eta_t.
\end{align*}
%The variance being one is confirmed per simulation.
\end{example}

Now we proceed to derive copulas for which the model in Eq.~\ref{Eq:Model(p,q)_updating_Eq} can recover ARCH and GARCH like processes. For the ARCH process we obtain a similar result to the one from \cite{dias2024garch}. However, for the GARCH process, where \cite{dias2024garch} use the \SP{infinite partial dependence} approach by \cite{mcneil2022time}, we use the model from Eq.~\ref{Eq:Model(p,q)_updating_Eq}, which is fundamentally different.

\begin{prop}[Recovering ARCH$(1)$, cf.~\cite{dias2024garch}]
\label{prop:ARCH(1)}
Let $\{W_t\}_{t\in\mathbb{Z}}$ be a copula-AR model with updating equation $W_t = C_{2|1}^{-1}(\varepsilon_t, W_{t-1})$ with $\varepsilon_t \stackrel{iid}{\sim} U(0,1)$. 
%\SPP{
Let $Z_t$ be an ARCH$(1)$ process with parameters $\alpha_0, \alpha_1 \geq 0$ and Gaussian innovations $\{\eta_t\}_{t\in\mathbb{Z}}$. Assume that $\{Z_t\}_{t\in\mathbb{Z}}$ is stationary, i.e., its parameters fulfill $\mathbb{E}(\log [\alpha_1 \eta_t^2]) < 0$ (compare Theorem 2.1 by \cite{francq2019garch}). Let $\Psi$ be the unconditional marginal distribution associated to this Gaussian ARCH$(1)$ process.
%} 
Then it holds that $\Psi^{-1}(W_t)$ follows a Gaussian ARCH$(1)$ process if
\begin{align*}
C_{21}(u_2, u_1) = \int_{0}^{u_1} \Phi\left[ \frac{\Psi^{-1}(u_2)}{\sqrt{\alpha_0 + \alpha_1 (\Psi^{-1}(u_1'))^2}} \right] \rmd u_1',
\end{align*}
where $(\alpha_0, \alpha_1)$ are the parameters of the stationary Gaussian ARCH$(1)$ process.
\end{prop}

\begin{proof}
A proof is given in Section 3.2 of \cite{dias2024garch}. In~\ref{App:Proof:ARCH(1)} we give an alternative proof.
\end{proof}

Before giving the result for the GARCH$(1,1)$-model, we consider a special case of the ARCH$(1)$ model. In this special case the copula in Prop.~\ref{prop:ARCH(1)} is the $t$-copula.

\begin{prop}[The $t$-copula is an ARCH-copula]
\label{prop:t-copula_ARCH(1)}
Let $\{\eta_t\}_{t\in\mathbb{Z}}$ be independent white noise with $\eta_t \sim t_{\nu+1}$. Let $Y_t = \sigma_t \eta_t$, where $\sigma_t^2 = \omega + \alpha Y_{t-1}^2$ and $\omega = \nu / (\nu + 1)$ as well as $\alpha = 1 / (\nu + 1)$ under the restriction $\nu > 0$. Then the time series $\{Y_t\}_{t\in\mathbb{Z}}$ is stationary and the stationary distribution is $t_{\nu}$. Furthermore, the copula of $(Y_t, Y_{t-1})$ is given as the $t$-copula with $\nu>0$ degrees of freedom and \SP{primary parameter} $\rho = 0$.
\end{prop}

\begin{proof}
The proof is given in \ref{App:Proof:t-copula_ARCH(1)}. It follows directly from the $t$-copulas conditional quantile function.
\end{proof}

\begin{remarks}
\label{Remarks:t-copula_ARCH(1)}
Remarks on Prop.~\ref{prop:t-copula_ARCH(1)}: The $t$-copula is an ARCH-copula.
\begin{itemize}
\item[a)] The $t$-copula has been used extensively for modeling volatile time series by e.g. \cite{simard2015forecasting}, \cite{loaiza2018time} and \cite{pappert2023forecasting}. This proposition provides a theoretical background for the usage of the $t$-copula for volatility modeling.
\item[b)] The set-up described in Prop.~\ref{prop:t-copula_ARCH(1)} is unlikely to be found in real time series. In particular it is unlikely that the copula and the marginal distributions have the same degree-of-freedom parameter $\nu$ if not enforced. Furthermore, the parameters $\omega$ and $\alpha$ are typically modeled independently. The relation $\nu \alpha = \omega$ is unlikely to hold in actual modeling. The proposition should be understood as a proof of concept.
\end{itemize}
\end{remarks}

\begin{prop}[Recovering GARCH$(1,1)$]
\label{prop:GARCH(1,1)}
Let $\{U_t\}_{t\in\mathbb{Z}}$ follow the model in Eq.~\ref{Eq:Model(p,q)_updating_Eq} with $p=q=1$, i.e., $U_t = K_{2|1}^{-1}(\varepsilon_t, W_{t-1})$, where $W_t = C_{2|1}^{-1}(\varepsilon_t, W_{t-1})$ and $\varepsilon_t \stackrel{iid}{\sim} U(0,1)$. 
%\SPP{
Let $\{Z_t\}_{t\in\mathbb{Z}}$ be a GARCH$(1,1)$ process with parameters $\alpha_0, \alpha_1, \beta_1 > 0$ and Gaussian innovations $\{\eta_t\}_{t\in\mathbb{Z}}$. Assume that $\{Z_t\}_{t\in\mathbb{Z}}$ is stationary, i.e., $\mathbb{E}(\log[\alpha_1 \eta_t^2 + \beta_1]) < 0$ (cf.~\cite{francq2019garch} Theorem 2.1). Let $\Psi$ be the unconditional marginal distribution associated to this stationary Gaussian-GARCH$(1,1)$ process.
%} 
Then it holds that $\Psi^{-1}(U_t)$ follows a Gaussian-GARCH$(1,1)$ process if
\begin{align*}
C_{21}(u_2, u_1) =& \int_{0}^{u_1} \Phi\left[ \sqrt{\frac{\Gamma^{-1}(u_2)^2}{\alpha_0 \alpha_1 + \alpha_1 \Gamma^{-1}(u_1')^2} - \frac{\beta_1}{\alpha_1}} \right] \rmd u_1',
\\
K_{21}(u_2, u_1) =& \int_0^{u_1} \Phi\left[ \frac{\Psi^{-1}(u_2)}{\sqrt{\alpha_0 + \Gamma^{-1}(u_1')^2}} \right] \rmd u_1'
\end{align*}
where $(\alpha_0, \alpha_1, \beta_1)$ are the parameters of the stationary Gaussian-GARCH$(1,1)$ process and $\Gamma$ is the CDF of the modified-ARCH process defined in the proof.
\end{prop}

\begin{proof}
The proof is given in \ref{App:Proof:GARCH(1,1)}.
\end{proof}

\begin{remarks}
\label{Remarks:GARCH(1,1)}
Remarks on Prop.~\ref{prop:GARCH(1,1)}: Recovering GARCH$(1,1)$
\begin{itemize}
\item[a)] The proposition regards recovering the Gaussian-GARCH$(1,1)$ process. The distribution of the innovations $\{\eta_t\}_{t\in\mathbb{Z}}$ is assumed to be Gaussian in this case. The generalization to other conditional distributions (e.g. $t$ or JSU) is straightforward by replacing $\Phi$ with a suitable distribution.
\item[b)] The proposition assures that a standard uniform process with GARCH$(1,1)$-dynamics can be constructed. This allows for processes with arbitrary stationary distributions and GARCH$(1,1)$ dynamics.
\end{itemize}
\end{remarks}
%
%\subsection{Properties of the model}
\subsection{Dependence properties of the model}

%Now 
In this subsection, we proceed to derive properties of the \SP{model in Eq.~\ref{Eq:Model(p,q)_updating_Eq}} and the MAG$(1)$-copula model. Before giving new results, we recall the properties already derived 
%\SPP{
in Sect. 3.14 of \cite{joe2014dependence}.
%} 
First, the joint distribution of $(U_t, \varepsilon_{t-1}, \hdots, \varepsilon_{t-q})$ in the $q$-dependent model is shown to be given by the MAG-copula $K$. Then the likelihood function (respectively, the joint density) is derived. Second, the likelihood of the whole model is derived for different combinations of model orders $p$ and $q$. The distribution of $U_t$ is standard uniform for all $t \in \mathbb{Z}$, \SP{if $g$ and $h$ in Eq.~\ref{Eq:Model(p,q)_updating_Eq} are appropriate copula conditional quantile functions.} Furthermore, the AR copula-based time series model is stationary and the MAG copula-based time series model is stationary and obviously $q$-dependent and hence ergodic. The full model, i.e., Eq.~\ref{Eq:Model(p,q)_updating_Eq}, is stationary since the unconditional distribution of $U_t$ is the same for all $t$ and the updating equation is not time dependent (in terms of Markov processes this property may also be called time-homogeneity). 

For the $1$-dependent process $\{V_t\}_{t\in\mathbb{Z}}$, following the updating equation
\begin{align*}
V_t = h(\varepsilon_t, \varepsilon_{t-1}), \quad \varepsilon_t \stackrel{iid}{\sim} U(0,1),
\end{align*}
the following holds. For $K$ being the (co-/counter-) monotonicity copula, $K = C^{+}$ or $K = C^{-}$, respectively, the process $V_t = \varepsilon_{t-1}$ or $V_t = -\varepsilon_{t-1}$ is recovered. For $K$ being the independence copula, $K = C^{\perp}$, we obtain $V_t = \varepsilon_t$, compare also Sect. 3.2 of \cite{pappert2024moving}.
%Now we start to \SP{state} additional properties. 
%First, we are interested in the MAG$(1)$ process. 
%Before giving the properties we state some conditions.
Next, we state additional properties after listing some conditions
for the results.
%
%
%

%\SPP{
\begin{condition} Conditions for properties of the MAG$(1)$-copula model
\begin{itemize}
    \item[(A)] As $u,v \rightarrow 0, \ K_{12}(u,v) = \mathcal{O}(\min\{u,v\})$.
    \item[(B)] As $u\rightarrow 0$, there exists $0 < M < 1$ such that $K_{2|1}(u|s) = \mathcal{O}(u)$ uniformly for all $s \in (M,1)$.
    \item[(C)] $K_{12}$ has tail-order $\kappa_{K_{12}} = 2$.
    \item[(D)] There exists $b: [0,M]^2 \rightarrow \mathbb{R}_+$ such that $K_{12}(u w_1, u w_2) / u^2$ converges uniformly to $b(w_1, w_2)$ for all $(w_1,w_2) \in [0,M]^2$.
    \item[(E)] $K_{2|1}(u|s) = \mathcal{O}(u)$ uniformly for $s \in (0,1)$
\end{itemize}
\end{condition}
%}
%
%
%

The following distributional properties hold.
\begin{prop}[Properties of MAG$(1)$]
\label{prop:Distributional_MAG(1)}
Let $\{V_t\}_{t\in\mathbb{Z}}$ be a time series with updating equation $V_t = h(\varepsilon_t, \varepsilon_{t-1})$ and $\varepsilon_t \stackrel{iid}{\sim} U(0,1)$, \SP{where $h(\varepsilon_t,\varepsilon_{t-1}) = K_{2|1}^{-1}(\varepsilon_t|\varepsilon_{t-1})$ and $K_{21}$ is the copula of $(V_t,\varepsilon_{t-1})$.}
%It holds that
\begin{itemize}
\item[a)] The time series $\{V_t\}_{t\in\mathbb{Z}}$ is stationary and ergodic.
\item[b)] The joint distribution of consecutive observations is given as (cf.~Appendix B.2 from \cite{pappert2024moving})
\begin{align}
F_{V_t,V_{t-1}}(u,v) = \int_0^1 K_{12}(K_{2|1}(u|s), v) \rmd s.
\label{Eq:MAG(1)_Joint_Distribution}
\end{align}
\item[c)] If $K$ is positive quadrant dependent (PQD), then $F_{V_t, V_{t-1}}$ is also PQD. If $K$ is \SP{negative quadrant dependent (NQD),} then $F_{V_t,V_{t-1}}$ is NQD.
\item[d)] If $K_{2|1}$ is stochasticly increasing (SI) in its second argument, then $V_t$ is SI in $V_{t-1}$.
\item[e)] The \SP{Spearman's} $\rho$ associated to the joint distribution $F_{V_t, V_{t-1}}$ is bounded in absolute value by $\frac{1}{2}$ (cf.~\cite{tao2025tridiagonal}).
\item[f)] The lower and upper tail dependence coefficients, $\lambda_l$ and $\lambda_u$, associated to $F_{V_t, V_{t-1}}$ are bounded in absolute value by $\frac{1}{2}$ (cf.~\cite{embrechts2016bernoulli}).
\item[g)] \SP{If (A) and (B) hold, the copula in Eq.~\ref{Eq:MAG(1)_Joint_Distribution} has lower tail dependence.
If (C), (D) and (E) hold, the copula of consecutive observations (Eq.~\ref{Eq:MAG(1)_Joint_Distribution}) has lower tail-order $\kappa_{\text{MAG}(1)} = 2$ and hence no lower tail dependence. 
}
%TODO Add Proof in Appendix
\item[h)] Let $\{\tilde{V}_t\}_{t\in\mathbb{Z}}$ be a time series with updating equation $\tilde{V}_t = h(\varepsilon_{t-1}, \varepsilon_t)$, i.e., the MAG$(1)$ updating equation with reversed order. Furthermore let the MAG$(1)$-copula be the Gaussian Copula. Then it holds that $(V_t,\hdots,V_{t-r}) \stackrel{d}{=} (\tilde{V}_t,\hdots,\tilde{V}_{t-r}), \ \forall r \in \mathbb{Z}$.
\end{itemize}
\end{prop}

\begin{proof}
The proof is given in~\ref{App:Proof:MAG(1)}
\end{proof}

\begin{remarks}%
\label{Remarks:Distributional_MAG(1)}
Remarks on the properties of the MAG$(1)$ model.
\begin{itemize}
\item[a)] Part h) of the proposition states that a Gaussian MAG$(1)$ process has two representations. In particular, the process $\{V_t\}_{t\in\mathbb{Z}}$, defined as $V_t = \Phi[ (1-\alpha^2)^{1/2} \Phi^{-1}(\varepsilon_{t-1}) + \alpha \Phi^{-1}(\varepsilon_{t})]$ with i.i.d.~innovations $\{\varepsilon_t\}_{t\in\mathbb{Z}}$ can be represented as $V_t = \Phi[ (1-\alpha^2)^{1/2} \Phi^{-1}(\eta_{t-1}) + \alpha \Phi^{-1}(\eta_{t})]$ with $i.i.d.$ innovations $\{\eta_t\}_{t\in\mathbb{Z}}$. From the proposition it is clear that $\eta_t = \varepsilon_{t-1}$ and $\eta_{t-1} = \varepsilon_t$. Since there is no temporal structure in the innovations, we are free to permute them and obtain the two representations. We see that there is a critical point at which the there is only one representation. This point is reached when $\alpha = (1-\alpha^2)^{1/2}$, which is solved by $\alpha = \SP{\pm} 1 / \sqrt{2}$. Furthermore, given the representation $\{V_t\}$ with parameter $\alpha$, the parameter of $\{\tilde{V}_t\}$ has to be chosen as $\alpha^{\dagger} = (1-\alpha^2)^{1/2}$ \SP{such that $\{\tilde{V}_t\}$ is a representation of $\{V_t\}$.} We may call $\alpha^{\dagger}$ the reciprocal to the value $\alpha$. This property is the pendant to the two representations of a Gaussian MA$(1)$ process. In Gaussian MA$(1)$ models with \SP{parameters $\sigma > 0$ and $\theta\ne0$,} the two representations are
\begin{align*}
Y_t =& \varepsilon_t + \theta \varepsilon_{t-1}, \quad \varepsilon_t \sim WN(0, \sigma^2)
\\
\tilde{Y}_t =& \tilde{\varepsilon}_t + \theta^{-1} \tilde{\varepsilon}_{t-1}, \quad \tilde{\varepsilon}_t \sim WN(0,\sigma^2 \theta^2),
\end{align*}
compare, for example, \cite{hamilton2020time}. While the equivalence of both representations is usually proven by showing that their autocovariance functions coincide, the property could also be shown by using the permutation argument we presented. The result from our proposition may be seen as a complement of the MA$(1)$ result. While the classic result holds for MA$(1)$ processes with elliptic white noise, our result holds for 
%any i.i.d.~innovation process.
a $U(0,1)$ innovation process.
The two representations give rise to a non-identifiability issue in MLE. We address this issue in Sect.~\ref{Subsec:Estimation} and investigate the likelihood by simulations in Sect.~\ref{Sec:Nums_Sims}.
\item[b)] The \SP{Spearman's} rho being bounded in absolute value by $\frac{1}{2}$ is similar to the property of classic MA$(1)$ processes. In fact we used the classic property in our proof. The statement generalizes the classic result to any $1$-dependent process. 
%\SPP{\sout{The intuition for this result becomes clear when considering the underlying process again: $V_t = K_{2|1}^{-1}(\varepsilon_t| \varepsilon_{t-1})$.}}
%
\item[c)] The absolute upper bound for the tail dependence coefficients being $\frac{1}{2}$ follows from a general argument on the form of tail dependence matrices by \cite{embrechts2016bernoulli}. However, numerical evaluations and estimations on simulations indicate a sharper bound for 'standard' copulas such as the Gumbel, $t$ or Clayton copula (see Sect.~\ref{Sec:Nums_Sims}). Furthermore we show that the tail dependence coefficients for the Fr{\'e}chet copula are bounded by $\frac{1}{4}$.
\SP{It appears that the MAG$(1)$ construction with an absolutely continuous copula $K_{21}$ might not be able to result in a time series with serial tail dependence. There could be other $1$-dependent time series constructions that can lead to a tail dependence coefficient between $\frac{1}{4}$ and $\frac{1}{2}$.}
%This poses the question if there is an additional assumption under which the bound $\frac{1}{2}$ by \cite{embrechts2016bernoulli} can be sharpened.
\end{itemize}
\end{remarks}
Now we proceed with two simple examples of MAG$(1)$ processes to illustrate Prop.~\ref{prop:Distributional_MAG(1)}, We consider the Fr{\'e}chet and the Frank copula as MAG-copulas and investigate the induced tail dependence on consecutive observations from the resulting time series.

\begin{example}[Fr{\'e}chet-MAG$(1)$ copula process]
\label{Example:Frechet-MAG(1)}
We consider \SP{a MAG$(1)$ process $\{V_t\}_{t\in\mathbb{Z}}$,} where the MAG-copula is given as the Fr{\'e}chet copula, which is a convex combination of the independence and comonotonicity copula:
\begin{align*}
K_{21}(u_1, u_2) =& (1-\alpha) C^{\perp}(u_1, u_2) + \alpha C^+(u_1, u_2),
\nonumber
\\
=& (1-\alpha) u_1 u_2 + \alpha \min\{u_1, u_2\},
\end{align*}
%Sometimes the countermonotoncity copula, $C^{-}$ is also added to the convex combination in the Fr{\'e}chet copula. Here we omit this component.
\SP{where $\alpha \in [0,1]$.} The conditional distribution function of $K_{12}$ is given as
\begin{align*}
\SP{K_{2|1}}(u_1|u_2) = (1-\alpha)u_1 + \alpha \mathds{1}\{u_2 \leq u_1\}
\end{align*}
almost everywhere. Now we calculate the implied tail dependence coefficients of the joint distribution of $(V_t, V_{t-1})$ using Eq.~\ref{Eq:MAG(1)_Joint_Distribution},
\begin{align*}
F_{V_t,V_{t-1}}(x_1,x_2) =& \int_0^1 \left[ \alpha \min\{K_{2|1}(x_1|s), x_2\} + (1 - \alpha) x_2 K_{2|1}(x_1|s) \right] \rmd s
\nonumber
\\
=& \int_0^1 (1 - \alpha)x_2[ (1-\alpha)x_1 + \alpha \mathds{1}\{s \leq x_1\}] \rmd s + \int_0^1 \alpha \min\{(1-\alpha)x_1 + \alpha \mathds{1}\{s \leq x_1\}, x_2\} \rmd s
\nonumber
\\
=& (1-\alpha)^2 x_1x_2\int_0^1 ds + (1-\alpha)\alpha x_2 \int_0^{x_1} \rmd s
\nonumber
\\
&+ \int_0^{x_1} \alpha \min\{(1-\alpha)x_1 + \alpha, x_2\}\rmd s + \int_{x_1}^1 \alpha \min\{(1-\alpha)x_1, x_2\}\rmd s
\nonumber
\\
=& (1-\alpha)^2 x_1 x_2 + (1-\alpha)\alpha x_1 x_2 + x_1 \alpha \min\{ (1-\alpha)x_1 + \alpha, x_2 \} + (1 - x_1) \alpha \min\{(1-\alpha)x_1,x_2\}
\end{align*}
\SP{Now we consider $\lim_{u\rightarrow 0} F_{V_t,V_{t-1}}(u,u) / u$. For small $u$ relative to $\alpha$, $F_{V_t,V_{t-1}}(u,u)$ simplifies to}
\begin{align*}
F_{V_t,V_{t-1}}(u,u) = (1-\alpha)^2 u^2 + (1-\alpha)\alpha u^2 + \alpha u^2 + (1-u)\alpha(1-\alpha)u.
\end{align*}
Then it is clear that $\lim_{u\rightarrow 0} F_{V_t,V_{t-1}}(u,u) / u = \alpha (1 - \alpha)$. This expression is maximal at $\alpha = \frac{1}{2}$ with a value of $\frac{1}{4}$. 
%\SPP{
Hence the MAG$(1)$-Fr{\'e}chet copula process has tail dependence in the joint distribution of consecutive observations. The upper tail dependence can be calculated in an analog way. Checking conditions (A) and (B) from part g) of Prop.~\ref{prop:Distributional_MAG(1)} also confirms that lower tail dependence can be expected. \SP{To confirm} condition (A) we see that as $u,v \rightarrow 0$ $K_{21}(u,v) \approx \alpha \min\{u,v\}$. For condition (B) we see that $K_{2|1}(u|s) = \mathcal{O}(u)$ for all $s \in (0,1)$.
%}
\end{example}
%

%\SPP{
\begin{example}[Frank-MAG$(1)$ copula process]
Let $K_{12}$ be given as the Frank copula with parameter \SP{$\delta \neq 0$,}
\begin{align}
K_{12}(u_1,u_2) = - \delta^{-1} \log \left( \frac{1 - e^{-\delta} - (1 - e^{-\delta u_2})(1 - e^{- \delta u_1})}{1 - e^{-\delta}} \right).
\end{align}
Then
\begin{align}
%K_{1|2}(u_1|u_2) =& \frac{e^{-\delta u_2}}{\frac{1 - e^{-\delta}}{1 - e^{-\delta u_1}} - (1 - e^{-\delta u_2})}
K_{1|2}(u_1|u_2) =& \frac{e^{-\delta u_2}}{(1 - e^{-\delta})/(1 - e^{-\delta u_1}) - (1 - e^{-\delta u_2})}
\\
\stackrel{u_1 \rightarrow 0}{\approx} & u_1 \delta e^{-\delta u_2}(1 - e^{-\delta}).
\end{align}
Furthermore the Frank copula has lower tail order $\kappa_{\text{Frank}} = 2$ and
\begin{align}
b(w_1, w_2) = \frac{\delta w_1 w_2}{1 - e^{-\delta}},
\end{align}
compare Example 2.13 by \cite{joe2014dependence}. Conditions (C), (D) and (E) from part g) of Prop.~\ref{prop:Distributional_MAG(1)} are fulfilled. The copula of consecutive observations hence has lower tail order $\kappa_{\text{Frank-MAG}(1)} = 2$ and no tail dependence.
\end{example}
%}
%
Now we proceed with the model from Eq.~\ref{Eq:Model(p,q)_updating_Eq} with model orders $p = 1$ and $q = 1$. In the following proposition, we again state distributional properties.
\begin{prop}
\label{prop:CICARMA(1,1)}
\SP{Let $\{U_t\}_{t\in\mathbb{Z}}$ follow the model from Eq.~\ref{Eq:Model(p,q)_updating_Eq} with $p=q=1$, i.e., $\{U_t\}_{t\in\mathbb{Z}}$ follows the updating equations $U_t = h(\varepsilon_t, W_{t-1})$, with $W_s = g(\varepsilon_s, W_{s-1})$ where $s,t\in \mathbb{Z}$ and $\varepsilon_t \stackrel{iid}{\sim} U(0,1)$. The mapping $g(\,\cdot\,,\,\cdot\,) = C^{-1}_{2|1}(\,\cdot\,|\,\cdot\,)$ is the conditional quantile function corresponding to the bivariate copulas $C_{21}$, which by this construction is the copula of $(W_s,W_{s-1})$ for all $s$. The mapping \SP{$h(\,\cdot\,,\,\cdot\,) = K^{-1}_{2|1}(\,\cdot\,|\,\cdot\,)$} is the conditional quantile function corresponding to the bivariate copulas $K_{21}$, which is the copula of $(U_t,W_{t-1})$ for all $t$.} The following holds
\begin{itemize}
\item[a)] \SP{If $\{W_s\}_{s\in \mathbb{Z}}$ is stationary and ergodic,} the time series $\{U_t\}_{t\in\mathbb{Z}}$ is stationary and ergodic.
\item[b)] \SP{Assuming $C_{21}$ is absolutely continuous with density $c_{21}$,} the joint distribution of $(U_t, U_{t-1})$ is given as
\begin{align}
F_{U_t, U_{t-1}}(x_1, x_2) = \int_0^1 \int_0^{C_{2|1}^{-1}\big(K_{2|1}(x_2|w_2) \big | w_2\big)} K_{2|1}(x_1|w_1) c_{21}(w_1, w_2) \rmd w_1 \rmd w_2
\label{Eq:CoInnoCo(1,1)_joint_distribution}
\end{align}
%\SPP{
\item[c)] \SP{Assume that $K_{12}$ is absolutely continuous and has copula density $k_{12}$.} The copula of the random vector $(U_t, W_{t-1}, W_{t}, U_{t+1})$ can be represented as a D-vine copula $C$ where $k_{12}$, which is the copula density of $(W_{t-1},U_t)$, as well as of $(W_t, U_{t+1})$, and $c_{12}$, which is the copula density of $(W_{t-1},W_{t})$, make up the first tree of the D-vine. 
%The other (conditional) pair-copulas are provided in the proof.
%\texttt{Should be summarized here}.
The copula of $(U_t, W_t|W_{t-1})$  is $C^{+}$ (comonotonicity),
the copula of $(W_{t-1}, U_{t+1}|W_{t}) $ is $C^{\perp}$ (independence) and 
the copula of $(U_t, U_{t+1} | W_{t-1}=w_{t-1}, W_{t}=w_{t})$ can be
taken as $C^\perp$ because $(U_t|W_{t-1}=w_{t-1},W_t=w_t)$ is degenerate.
\item[d)] From c), it follows that the copula of $(U_t, W_{t-1}, W_{t}, U_{t+1})$ and $(U_t,U_{t+1})$ have lower (upper) tail dependence if $K_{12}$ and $C_{12}$ have continuous second order derivatives, have support on $[0,1]^2$ and are lower (upper) tail dependent.
%}
\end{itemize}
\end{prop}

\begin{proof}
The proof is given in~\ref{App:Proof:CICARMA(1,1)}.
\end{proof}

\begin{remarks}
\label{Remarks:Properties_CoInnoCo(1,1)}
Remarks on properties of the \SP{model from Eq.~\ref{Eq:Model(p,q)_updating_Eq}} with $p=q=1$.
\begin{enumerate}[label=\alph*)]
%\SPP{
\item Based on numerical calculations (via Gauss-Legendre quadrature) of
Spearman's rho based on Eq.~\ref{Eq:CoInnoCo(1,1)_joint_distribution} when $C_{12}$ and $K_{12}$ are in the same parametric family, the Spearman correlation  of lag one gets weaker than the one of a Markov$(1)$ process with copula $C_{12}$, when the dependence in $K_{12}$ is weaker than the dependence in $C_{12}$. On the other hand, the Spearman correlation of lag one gets stronger than the one of a Markov$(1)$ model with copula $C_{12}$ when the dependence in $K_{12}$ is stronger than the dependence of $C_{12}$, with a peak lag one correlation when Spearman's rho of $K_{12}$ is quite large. For Gaussian copulas, this result can also be derived analytically since the autocorrelation function is explicitly available 
%There, we observe a similar behavior, 
--- compare the comments at the end of Example~\ref{Example:Gaussian(1,1)}.
\item Part c) and d) of the proposition provide an easy way to construct a model which exhibits tail dependence, namely by choosing $K_{12}$ and $C_{12}$ as prescribed by part d). We note that this result matches with results from simulation/integration studies.
%}
\end{enumerate}
\end{remarks}

\begin{prop}
\label{prop:Joint_Distributions_higher_orders}
\SP{Let $\{U_t\}_{t\in\mathbb{Z}}$ follow the model from Eq.~\ref{Eq:Model(p,q)_updating_Eq} with AR copula $K_{(p+1),\hdots,1}$, MAG copula $C_{(q+1),\hdots,1}$, latent standard uniform process $\{W_s\}_{s\in\mathbb{Z}}$, $i.i.d.$ standard uniform innovations $\{\varepsilon_t\}_{t\in\mathbb{Z}}$ and model orders $p$ and $q$.} Then it holds that
%the joint distribution of consecutive observations $(U_t, U_{t-1})$ is given as
\begin{itemize}
\item[a)] \SP{For $p=1,q=2$, where $C_{21}$ is the copula of $(W_s,W_{s-1})$ for all $s$ and $K_{321}$ is the copula of the random vector $(U_t,\varepsilon_{t-1},W_{t-1})$ for all $t$, the joint distribution of consecutive observations is given as}
\begin{align*}
F_{U_t,U_{t-1}}(x_1, x_2) =& \int_{(0,1)^3} K_{3|21}(x_1|e_1,w_2) \times c_{12}(w_2, w_3)
\\
&\times \mathds{1}\{ e_1 \leq K_{3|21}(x_2|C_{2|1}(w_2|w_3), w_3) \}\rmd w_2 \rmd w_3 \rmd e_1,
\end{align*}
\SP{assuming $C_{12}$ is absolutely continuous.}
\item[b)] \SP{For $p=2,q=1$, where $C_{321}$ is the copula of $(W_s,W_{s-1},W_{s-2})$ for all $s$ and $K_{21}$ is the copula of $(U_t,W_{t-1})$ for all $t$, the joint distribution of consecutive observations is given as}
\begin{align*}
F_{U_t,U_{t-1}}(x_1, x_2) =& \int_{(0,1)^3} K_{2|1}(x_1|w_1) \times c_{321}(w_1, w_2, w_3)
\\
& \times \mathds{1}\{ w_1 \leq C_{3|21}^{-1}(K_{2|1}(x_2|w_2)|w_2, w_3) \} \rmd w_1 \rmd w_2 \rmd w_3,
\end{align*}
\SP{assuming $C_{321}$ is absolutely continuous.}
\item[c)] \SP{For $p=2,q=2$, where $C_{321}$ is the copula of $(W_s,W_{s-1},W_{s-2})$ for all $s$ and $K_{321}$ is the copula of $(U_t,\varepsilon_{t-1}, W_{t-2})$ for all $t$, the joint distribution of consecutive observations is given as}
\begin{align*}
F(x_1,x_2) =& \int_{(0,1)^4} K_{3|21}(x_1|e_1, w_2) \times c_{321}(w_2, w_3, w_4)
\\
& \times \mathds{1}\{ e_1 \leq K_{3|21}(x_2 | C_{3|21}(w_2|w_3,w_4), w_3) \} \rmd e_1 \rmd w_2 \rmd w_3 \rmd w_4,
\end{align*}
\SP{assuming $C_{321}$ is absolutely continuous.}
\end{itemize}
\end{prop}

\begin{proof}
The proof is given in \ref{App:Proof:Joint_Distribution_higher_orderss}.
\end{proof}

Now we \SP{state the likelihood of the model from Eq.~\ref{Eq:Model(p,q)_updating_Eq}} and provide an algorithm that calculates the likelihood iteratively.

\subsection{Estimation}
\label{Subsec:Estimation}

In this section we discuss the maximum likelihood estimation \new{(MLE).} 
\SP{ Sect.~3.14 of \cite{joe2014dependence} provides a general result for the likelihood, with a Jacobian and specific details for $(p,q) \in \{(0,1), (0,2), (1,2), (2,1)\}$} We note that estimation properties of the MAG$(1)$ process are also discussed by \cite{pappert2024moving}.\\

In order to derive the likelihood,
%\SP{given candidate copulas $C_{(p+1),\hdots,1}$ and $K_{(q+1),\hdots,1}$,}
the conditional density of $U_t|\mathcal{F}_{t-1}$ is necessary, where $\mathcal{F}_{s}$ denotes all information up to and including $s \in \mathbb{Z}$. \new{Assuming $K_{(q+1),\hdots,1}$ is absolutely continuous, the conditional density of $U_t|\mathcal{F}_{t-1}$ from the model in Eq.~\ref{Eq:Model(p,q)_updating_Eq} with arbitrary orders $p$ and $q$ is given as}
\begin{align*}
f_{U_t|\mathcal{F}_{t-1}}(x) = k_{(q+1),\hdots,1}(x,\varepsilon_{t-1},\hdots,\varepsilon_{t-q+1}, W_{t-q}),
\end{align*}
where $k_{(q+1),\hdots,1}$ is the $(q+1)$-variate density \SP{of the MAG-copula,} $(\varepsilon_{t-1},\hdots,\varepsilon_{t-q+1})$ are the \SP{innovations} and $W_{t-q}$ the latent AR-process. Given a set of candidate AR and MAG copulas, which we may denote by $\hat{K}_{(q+1),\hdots,1}$ and $\hat{C}_{(p+1),\hdots,1}$, with candidate parameter values, for which the likelihood ought to be estimated, the innovation at time $s$ can be estimated as $\hat{\varepsilon}_s = \hat{K}_{(q+1)|q,\hdots,1}(U_s|\hat{\varepsilon}_{s-1},\hdots,\hat{\varepsilon}_{s-q+1},\hat{W}_{t-q})$, given \SP{the parameters and} preceding values. Similarly \SP{the latent AR-process $\hat{W}_s$ at time $s$ can be estimated as $\hat{W}_{s} = \hat{C}_{(p+1)|p,\hdots,1}^{-1}(\hat{\varepsilon}_s|\hat{W}_{s-1},\hdots,\hat{W}_{s-p})$ given the parameters and preceding values.} It is evident that $\{\hat{\varepsilon}_s\}$ and $\{\hat{W}_s\}$ need to be calculated iteratively. In algorithm~\ref{Algo:PARMA_nll} we show how the likelihood can be calculated using the \texttt{rvinecopulib}-package in \texttt{R} \SP{(see \cite{nagler2022rvinecopulib}).} Given the algorithm for calculating the likelihood, numerical optimization can be utilized to find the minimum of the negative log-likelihood. In order to establish consistency, the stationarity and ergodicity of \SP{the estimated latent processes $\{\hat{\varepsilon}_s\}$ and $\{\hat{W}_s\}$ is needed for all feasible parameter values. Assuming that the data-generating process is the model from Eq.~\ref{Eq:Model(p,q)_updating_Eq}, the stationarity and ergodicity of the observed process $\{U_t\}_{t\in\mathbb{Z}}$ follows. If the copulas are chosen such that $\{\hat{W}_t\}$ fulfills the conditions from Theorem 4.40 by \cite{douc2014nonlinear}, the stationarity and ergodicity of $\{\hat{W}_t\}$ and $\{\hat{\varepsilon}_t\}$ follows.} Then, under additional assumptions, Theorem 4.2 from \cite{wooldridge1994estimation} can be applied to establish the uniform weak law of large numbers for the log-likelihood estimator. Then Theorem 4.3 by \cite{wooldridge1994estimation} can be applied to establish consistency (again under additional assumptions). Establishing stationarity and ergodicity of the estimated latent processes seems to be the crucial step. Currently it not clear how these conditions can be translated to conditions on the AR and MAG copulas.\\

\vspace*{1cm}
\begin{minipage}{\textwidth}
\centering
\begin{algorithm}[H]
\caption{likelihood calculation for the \SP{model from Eq.~\ref{Eq:Model(p,q)_updating_Eq}} based on the \texttt{R}-package \texttt{rvinecopulib}. \SP{The objects \texttt{model{\_}AR} and \texttt{model{\_}MA} contain all information about the AR copula $C_{(p+1),\hdots,1}$ and the MAG copula $K_{(q+1),\hdots,1}$ respectively.}}\label{Algo:PARMA_nll}
\SetKwInOut{Input}{input}
\SetKwInOut{Output}{output}
\Input{$x$; \texttt{model{\_}AR}, \texttt{model{\_}MA}, \SP{$p$ and $q$}}
\Output{Negative Log-Likelihood}
$l \gets \texttt{length(x)}$\\
$d, e, W \gets \texttt{vector()}$ \\
$r \gets \max(p,q)$\\
$d_{1:r}, e_{1:r}, W_{1:r} \gets (1,0.5,0.5)$\\
\For{$t\gets (r+1)$ \KwTo $l$}{
	$e^{\dagger} \gets \texttt{inverse{\_}rosenblatt}(u = \left(W_{t-q}, e_{(t-q+1):(t-1)}, 0.5\right), \texttt{model{\_}MA})$\\
	$e_t \gets \texttt{rosenblatt}(x = (e^{\dagger}_{1:q},x_t), \texttt{model{\_}MA})_{(q+1)}$ \\
	$W^{\dagger} \gets \texttt{rosenblatt}(x = (W_{(t-p):(t-1)}, 0.5), \texttt{model{\_}AR})$\\
    $W_{t} \gets \texttt{inverse{\_}rosenblatt}((W^{\dagger}_{1:p}, e_t), \texttt{model{\_}AR})_{(p+1)}$\\
    $d_t \gets \texttt{dvinecop}(u = (x_t, e_{(t-1):(t-q+1)}, W_{(t - q)}), \texttt{model{\_}MA})$
    }
    $Q \gets \sum_{t = 1}^l [-\log d_t]$\\
    \Return $Q$
\end{algorithm}
\end{minipage}

\section{Simulation and Numerical Experiments}
\label{Sec:Nums_Sims}

\SP{In this section we evaluate dependence measures of consecutive observations from MAG$(1)$-processes. Furthermore, we investigate identifiability issues of the MAG$(1)$-process.}
\subsection{Dependence of consecutive observations from a MAG$(1)$-process}
We evaluate dependence measures of \SP{consecutive observations} from the model from Eq.~\ref{Eq:Model(p,q)_updating_Eq} for lag orders $(p,q) \in \{ (0,1), (1,1)\}$. \SP{In particular we are interested in Spearman's rho, the lower tail dependence coefficient and the corresponding tail order.} We refer to Sect.~2.13 and 2.16 of \cite{joe2014dependence} for tail dependence and the tail order, respectively. We use the joint distributions derived in the respective propositions, cf.~Prop.~\ref{prop:Distributional_MAG(1)} and~\ref{prop:CICARMA(1,1)}, and evaluate the integrals numerically.
% \SP{and by Monte-Carlo simulation.} 
For the numerical approximation of integrals, we use Gaussian quadrature. 
\SP{The lower tail dependence coefficient is approximated by evaluating $\Lambda(u) = F_{U_t,U_{t-1}}(u,u) / u$ for different $u$-values. The tail order is approximated by evaluating $\mathbb{K}(u) = \log F_{U_t, U_{t-1}}(u,u)] / \log u$. It holds that $\lim_{u \rightarrow 0}\Lambda(u) = \lambda_L$ and $\lim_{u \rightarrow 0} \mathbb{K}(u) = \kappa_L$. The limiting values can not be evaluated numerically, however, evaluating the expressions at decreasing sequences of $u$, can give an idea of what the limiting value is. In particular we evaluate the above expressions for $u \in \{0.1, 0.05, 0.03, 0.01, 0.005\}$. If $\mathbb{K}(u)$ is increasing as $u \rightarrow 0$, that is a sign of no tail dependence. We note that the rate of convergence can be slow, hence the results need to be interpreted with caution. To this end we take guidance from the results, where the true tail dependence is known (the Gaussian and Frank MAG$(1)$ processes).
The results for different MAG$(1)$-copulas (Frank, Gumbel, reflected Gumbel and Gaussian) are reported in Tables~\ref{Table:MAG(1)_tail_Frank}--\ref{Table:MAG(1)_tail_Gaussian}. 
The copula parameter values are chosen such that they correspond to Spearman's rho values of $\rho_{S} \in \{0.1, 0.2, \ldots, 0.8, 0.9\}$.} 
%
% --- MAG(1) Tables interpretation ---
%

\SP{%
Table~\ref{Table:MAG(1)_tail_Frank} displays the tail coefficient and order for the Frank-MAG$(1)$ process. We know that the tail order of consecutive observations from the Frank-MAG$(1)$ process is $\kappa_L = 2$ and consequently $\lambda_L = 0$. This is reflected in the approximations. We note, however, that the values for $\mathbb{K}(0.0005)$ still deviate substantially from $\kappa_L = 2$, especially for larger parameter values. This is most likely due to the slow convergence rates as $u \rightarrow 0$, as mentioned before. In Table~\ref{Table:MAG(1)_tail_Gaussian}, the consecutive-observation tail dependence measures are displayed for the Gaussian-MAG$(1)$ process. Here the true tail order is also known, and in particular it holds that $\kappa_L > 1$ and hence there is no tail dependence. Nevertheless, we observe that $\mathbb{K}$ decreases as $u$ decreases. Unexpectedly the tail order may also decrease as $u \rightarrow 0$ despite there being no tail dependence. This realization helps contextualizing the results for other copulas. For the reflected Gumbel-MAG$(1)$ process, the tail order $\kappa_L$ is not known. The approximations in Table~\ref{Table:MAG(1)_tail_180Gumbel} suggest that there is no tail dependence. The tail order is decreasing as $u$ decreases. The pattern is similar to the Gaussian-MAG$(1)$ process. The same holds for the Gumbel-MAG$(1)$ process in Table~\ref{Table:MAG(1)_tail_Gumbel}. We note that for the reflected Gumbel-MAG$(1)$, a potential corresponding theoretical lower bound for $\kappa_L$ is $(2-1/\theta) \in \{1.065, 1.138, 1.206, 1.275, 1.351, 1.429, 1.517, 1.612, 1.732\}$, respectively. Furthermore, we note that the Clayton-MAG$(1)$ process can be expected to behave similar to the reflected Gumbel-MAG$(1)$ in the lower tail. The reflected Clayton-MAG$(1)$ can be expected to behave similar to the Frank-MAG$(1)$ process in terms of lower tail dependence. In particular, Proposition~\ref{prop:Distributional_MAG(1)} g) applies. Results from the threshold and tail-weighted dependence measure methods by \cite{krupskii2019nonparametric} and \cite{lee2018tail} qualitatively match the results based on integration but are not as accurate.
}

\subsection{Identifiability of MAG$(1)$ MLE}

Another research question, which we investigate by simulation, is related to part h) of Proposition~\ref{prop:Distributional_MAG(1)}. Recall that it is shown that for a Gaussian-MAG$(1)$ process $\{V_t\}_{t\in\mathbb{Z}}$, there is a process $\{\tilde{V}_t\}_{t\in\mathbb{Z}}$ which is equal in distribution to $\{V_t\}_{t\in\mathbb{Z}}$. 
This result in the process $\{V_t\}$ having two representations. We explore how these to representations affect MLE. To this end we simulate a copula-MAG$(1)$ time series, calculate the likelihood for different parameter values and check whether the minimum of the \SP{negative log-likelihood (NLL)} is at the true parameter value, which \SP{is} used in simulation. First, we check the behavior of the NLL of the Gaussian-MAG$(1)$. In Fig.~\ref{Fig:nll_Gausian_MAG(1)}, the estimated NLL is plotted \SP{against} the copula parameter for three different values of the true parameter. 
The three different values are $\alpha\in\{0.25, 0.5, 0.875\}$. The true value is marked by the red dashed line. Additionally, the blue dashed line marks the reciprocal value to the true parameter value. This is the parameter value that needs to be chosen in the construction of $\{\tilde{V}_t\}_{t\in\mathbb{Z}}$ such that it is a representation of $\{V_t\}_{t\in\mathbb{Z}}$, \SP{compare Remark~\ref{Remarks:Distributional_MAG(1)}, part a).} 
For the Gaussian copula the reciprocal value can be calculated as $\alpha^{\dagger} = \sqrt{1 - \alpha^2}$. We also recall that the critical value for the two representations of the Gaussian-MAG$(1)$ process is $1 / \sqrt{2} \approx 0.7071$. Given these information, the plots in Fig.~\ref{Fig:nll_Gausian_MAG(1)} can be interpreted. First, for $\alpha \in \{0.25, 0.5\}$ we observe that the minimum of the NLL is attained at the true parameter value. 
This is expected. However, for $\alpha = 0.875$, which is above the critical threshold $1 / \sqrt{2}$ we observe that the minimum of the NLL is attained at the reciprocal value. Since the MAG$(1)$ model with the reciprocal value is just a different representation of the \SP{data-generating process,} the MLE does not give rise to bias. We conclude that the NLL is minimized by values which lie below the critical value. 
The reason for this behavior can be found when considering Proposition 7 by \cite{pappert2024moving}. The proposition concerns the consistency of MLE of the MAG$(1)$ process. A crucial step in the derivation is that \SP{the time series of estimated innovations $\{\hat{e}_t\}_{t\in\mathbb{Z}}$ which is defined by the updating equation, $\hat{e}_t = \hat{K}_{2|1}(V_t|\hat{e}_{t-1})$, where $\{V_t\}$ is the observed process, is stationary and ergodic.}
\cite{pappert2024moving} assures this by assuming that $\hat{K}_{2|1}$ induces a contraction-on-average in the sense of Theorem 4.40 by \cite{douc2014nonlinear}. For the Gaussian copula to induce a contraction-on-average, the parameter has to fulfill $|\hat{\alpha}| < 1 / \sqrt{2}$. \SP{(This can be seen by considering the stochastic recurrence equation function $\tilde{f}_z(x) = \Phi^{-1}(\hat{K}_{2|1}(z|\Phi(x))$. This mapping induces a contraction if the parameter of $\hat{K}_{2|1}$, say $\hat{\alpha}$, fulfills $|\hat{\alpha} / (1 - \hat{\alpha}^2)^{1/2}| < 1$. This is fulfilled if $|\hat{\alpha}| < 1 / \sqrt{2}$.)} If \SP{$\hat{\alpha}$} is outside the bound, then the estimated innovation time series is not stationary and ergodic. This is exemplified in the left plot of Fig.~\ref{Fig:MAG(1)_residuals}. 
After a burn-in period, the time series only oscillates between $0$ and $1$. \SP{(To illustrate how exactly the oscillation arises consider a $K_{12}$ with strong dependence. Then $\hat{e}_t = K_{2|1}(\hat{V}_t|\hat{e}_{t-1}) \approx \mathds{1}_{\hat{V}_t \geq \hat{e}_{t-1}}$ is either very close to $0$ or to $1$. If $\hat{e}_{t-1}$ is close to 0, then $\hat{e}_t$ can be expected to be close to $1$. If $\hat{e}_{t-1}$ is close to $1$, $\hat{e}_t$ can be expected to be close to $0$. Compare also \cite{hua2014strength} for the behavior of boundary CDFs associated to copulas.)}
For the reciprocal of the true value, the residuals are stationary and ergodic. This is exemplified in the right plot of Fig.~\ref{Fig:MAG(1)_residuals}. Notice that the critical value for \SP{$\{\hat{e}_t\}_{t\in\mathbb{Z}}$} to be stationary and ergodic coincides with the critical value of the two representations of the Gaussian-MAG$(1)$ process. We note that this behavior of the residuals and likelihood is similar to the behavior of the classical Gaussian MA$(1)$. As discussed in Remark~\ref{Remarks:Distributional_MAG(1)}, part a), this model also has two representations. For estimation purposes the parameter is restricted to the \textit{invertible region} $(-1,1)$ (compare Sect. 3.7 by \cite{hamilton2020time}). When relating the Gaussian-MAG$(1)$ to the classical Gaussian MA$(1)$ by $\Phi$-quantile transforming, see Example~\ref{Example:Gaussian_MA(1)}, we see that the relation between the MAG-parameter $\beta_1$ and the MA-parameter $\theta_1$ is given as $\theta_1 = \beta_1 / (1-\beta_1^2)^{1/2}$. Plugging in the critical value for $\beta_1$ we see that \SP{$\theta_1 = \pm 1$} is recovered which is the critical value of the classical MA$(1)$ process. These findings imply that the MLE can not be consistent for $|\hat{\alpha}| \geq 1 / \sqrt{2}$. This in in line with Proposition 7 by \cite{pappert2024moving}. In order to assure consistency, the parameter space may be restricted accordingly. This is not a problem, since $1 / \sqrt{2}$ is also the critical value for the two representations. \SP{Hence a Gaussian-MAG$(1)$ with estimated $|\hat{\alpha}| < 1 / \sqrt{2}$ is not restricted in its temporal dependence by the bounds on $\hat{\alpha}$.} \\
We also investigate the NLL for the Gumbel-MAG$(1)$ process. In Fig.~\ref{Fig:nll_Gumbel_MAG(1)} we again show the estimated NLL in dependence of parameter values. The red-dashed line again marks the true value and the blue-dashed line an approximate reciprocal value. The approximate reciprocal value is calculated as follows: We calculate the Kendall's tau value implied by the parameter of the Gumbel copula $\alpha$, then we calculate the parameter of a Gaussian copula with such Kendall's tau value. Denote this value by $\rho^{*}$. We then calculate $\rho^{\dagger} = \sqrt{1 - \rho^{*}}$, obtain the Kendall's tau value implied by this parameter and last calculate the parameter of the Gumbel copula, which has the same Kendall's tau value. This is only an approximation. In fact it is not even proven that there are two representations for the Gumbel-MAG$(1)$ process. But to explore the process, we make use of the approximation. Now for $\alpha = 1.25$, the NLL is minimized at the true value. For $\alpha = 2.1$, the NLL attains its minimum at the approximate reciprocal value. However, for $\alpha = 4.5$ neither the true value nor the approximate reciprocal value minimize the likelihood, although the approximate reciprocal value is close. \SP{The results indicate that there may be identifiability for Gumbel-MAG$(1)$ process, unlike with the Gaussian copula.}

\newpage
\begin{table}
\begin{center}
\caption{Dependence measures of consecutive observations from the Frank-MAG$(1)$ process.
The copula parameters $\theta$ are chosen to correspond to to Spearman rho values $\rho_{S} \in \{0.1, 0.2, \ldots, 0.8, 0.9\}$.} 
\label{Table:MAG(1)_tail_Frank}
\resizebox{0.99\columnwidth}{!}{%
\begin{tabular}{c|ccccc|ccccc|cc} \hline
$\theta$ & $\Lambda(0.01)$ & $\Lambda(0.005)$ & $\Lambda(0.003)$ & $\Lambda(0.001)$ & $\Lambda(0.0005)$ & $\mathbb{K}(0.01)$ & $\mathbb{K}(0.005)$ & $\mathbb{K}(0.003)$ & $\mathbb{K}(0.001)$ & $\mathbb{K}(0.0005)$ & $\kappa_L$ & $\rho_S$ \\
\hline
0.60 & 0.013 & 0.007 & 0.004 & 0.001 & 0.001 & 1.939 & 1.947 & 1.951 & 1.959 & 1.963 & 2 & 0.098 \\
1.22 & 0.017 & 0.009 & 0.005 & 0.002 & 0.001 & 1.884 & 1.898 & 1.906 & 1.921 & 1.928 & 2 & 0.194 \\
1.88 & 0.022 & 0.011 & 0.007 & 0.002 & 0.001 & 1.832 & 1.852 & 1.864 & 1.885 & 1.895 & 2 & 0.282 \\
2.61 & 0.027 & 0.014 & 0.008 & 0.003 & 0.001 & 1.782 & 1.808 & 1.823 & 1.851 & 1.864 & 2 & 0.360 \\
3.45 & 0.034 & 0.017 & 0.011 & 0.004 & 0.002 & 1.734 & 1.765 & 1.784 & 1.817 & 1.833 & 2 & 0.421 \\
4.47 & 0.042 & 0.022 & 0.013 & 0.004 & 0.002 & 1.688 & 1.722 & 1.744 & 1.783 & 1.802 & 2 & 0.462 \\
5.82 & 0.053 & 0.028 & 0.017 & 0.006 & 0.003 & 1.640 & 1.677 & 1.702 & 1.746 & 1.769 & 2 & 0.475 \\
7.90 & 0.066 & 0.036 & 0.022 & 0.008 & 0.004 & 1.589 & 1.627 & 1.654 & 1.704 & 1.729 & 2 & 0.450 \\
12.26 & 0.086 & 0.050 & 0.033 & 0.012 & 0.006 & 1.532 & 1.564 & 1.590 & 1.643 & 1.673 & 2 & 0.364 \\
\hline
\end{tabular}
}
\end{center}    
\end{table}
\begin{table}
\begin{center}
\caption{Dependence measures of consecutive observations from the Gumbel-MAG$(1)$ process}
\label{Table:MAG(1)_tail_Gumbel}
\resizebox{0.99\columnwidth}{!}{%
\begin{tabular}{c|ccccc|ccccc|c} \hline
$\theta$ & $\Lambda(0.01)$ & $\Lambda(0.005)$ & $\Lambda(0.003)$ & $\Lambda(0.001)$ & $\Lambda(0.0005)$ & $\mathbb{K}(0.01)$ & $\mathbb{K}(0.005)$ & $\mathbb{K}(0.003)$ & $\mathbb{K}(0.001)$ & $\mathbb{K}(0.0005)$ & $\rho_S$ \\
\hline
1.07 & 0.015 & 0.008 & 0.005 & 0.002 & 0.001 & 1.912 & 1.912 & 1.911 & 1.911 & 1.911 & 0.096 \\
1.16 & 0.023 & 0.013 & 0.009 & 0.003 & 0.002 & 1.819 & 1.819 & 1.819 & 1.819 & 1.819 & 0.198 \\
1.26 & 0.033 & 0.020 & 0.014 & 0.006 & 0.004 & 1.738 & 1.738 & 1.738 & 1.737 & 1.737 & 0.284 \\
1.38 & 0.047 & 0.030 & 0.021 & 0.010 & 0.007 & 1.664 & 1.663 & 1.663 & 1.662 & 1.662 & 0.358 \\
1.54 & 0.065 & 0.044 & 0.033 & 0.017 & 0.012 & 1.592 & 1.591 & 1.590 & 1.588 & 1.587 & 0.420 \\
1.75 & 0.087 & 0.061 & 0.047 & 0.027 & 0.019 & 1.530 & 1.527 & 1.525 & 1.522 & 1.521 & 0.460 \\
2.07 & 0.111 & 0.082 & 0.066 & 0.041 & 0.030 & 1.477 & 1.472 & 1.469 & 1.464 & 1.461 & 0.473 \\
2.58 & 0.128 & 0.098 & 0.081 & 0.054 & 0.042 & 1.447 & 1.438 & 1.433 & 1.423 & 1.418 & 0.445 \\
3.73 & 0.119 & 0.094 & 0.079 & 0.055 & 0.044 & 1.462 & 1.447 & 1.438 & 1.421 & 1.412 & 0.355 \\
\hline
\end{tabular}
}
\end{center}    
\end{table}

\begin{table}
\begin{center}
\caption{Dependence measures of consecutive observations from the reflected Gumbel-MAG$(1)$ process}
\label{Table:MAG(1)_tail_180Gumbel}
\resizebox{0.99\columnwidth}{!}{%
\begin{tabular}{c|ccccc|ccccc|c} \hline
$\theta$ & $\Lambda(0.01)$ & $\Lambda(0.005)$ & $\Lambda(0.003)$ & $\Lambda(0.001)$ & $\Lambda(0.0005)$ & $\mathbb{K}(0.01)$ & $\mathbb{K}(0.005)$ & $\mathbb{K}(0.003)$ & $\mathbb{K}(0.001)$ & $\mathbb{K}(0.0005)$ & $\rho_S$ \\
\hline
1.07 & 0.092 & 0.086 & 0.083 & 0.078 & 0.076 & 1.518 & 1.463 & 1.429 & 1.369 & 1.339 & 0.096 \\
1.16 & 0.157 & 0.146 & 0.138 & 0.125 & 0.117 & 1.402 & 1.364 & 1.341 & 1.302 & 1.283 & 0.198 \\
1.26 & 0.189 & 0.170 & 0.157 & 0.132 & 0.119 & 1.361 & 1.335 & 1.319 & 1.293 & 1.281 & 0.284 \\
1.38 & 0.194 & 0.167 & 0.149 & 0.115 & 0.098 & 1.356 & 1.338 & 1.328 & 1.313 & 1.306 & 0.358 \\
1.54 & 0.176 & 0.142 & 0.121 & 0.085 & 0.068 & 1.377 & 1.368 & 1.363 & 1.356 & 1.353 & 0.420 \\
1.75 & 0.144 & 0.109 & 0.089 & 0.057 & 0.043 & 1.421 & 1.418 & 1.416 & 1.415 & 1.415 & 0.460 \\
2.07 & 0.106 & 0.075 & 0.058 & 0.033 & 0.023 & 1.487 & 1.488 & 1.489 & 1.492 & 1.494 & 0.473 \\
2.58 & 0.072 & 0.047 & 0.035 & 0.018 & 0.012 & 1.572 & 1.576 & 1.578 & 1.584 & 1.581 & 0.445 \\
3.73 & 0.041 & 0.025 & 0.016 & 0.008 & 0.004 & 1.692 & 1.697 & 1.710 & 1.698 & 1.735 & 0.355 \\
\hline
\end{tabular}
}
\end{center}    
\end{table}

\begin{table}
\begin{center}
\caption{Dependence measures of consecutive observations from the Gaussian-MAG$(1)$ process}
\label{Table:MAG(1)_tail_Gaussian}
\resizebox{0.99\columnwidth}{!}{%
\begin{tabular}{c|ccccc|ccccc|cc} \hline
$\rho$ & $\Lambda(0.01)$ & $\Lambda(0.005)$ & $\Lambda(0.003)$ & $\Lambda(0.001)$ & $\Lambda(0.0005)$ & $\mathbb{K}(0.01)$ & $\mathbb{K}(0.005)$ & $\mathbb{K}(0.003)$ & $\mathbb{K}(0.001)$ & $\mathbb{K}(0.0005)$ & Kaplag1 & rhoSlag1 \\
\hline
0.105 & 0.020 & 0.011 & 0.007 & 0.003 & 0.002 & 1.852 & 1.849 & 1.847 & 1.844 & 1.842 & 1.811 & 0.099 \\
0.209 & 0.035 & 0.022 & 0.015 & 0.007 & 0.004 & 1.730 & 1.725 & 1.721 & 1.715 & 1.712 & 1.661 & 0.196 \\
0.313 & 0.055 & 0.037 & 0.027 & 0.015 & 0.010 & 1.630 & 1.624 & 1.619 & 1.612 & 1.608 & 1.542 & 0.285 \\
0.416 & 0.079 & 0.056 & 0.044 & 0.026 & 0.018 & 1.552 & 1.544 & 1.539 & 1.531 & 1.527 & 1.451 & 0.363 \\
0.518 & 0.103 & 0.076 & 0.061 & 0.039 & 0.029 & 1.493 & 1.485 & 1.480 & 1.471 & 1.467 & 1.386 & 0.426 \\
0.618 & 0.122 & 0.093 & 0.076 & 0.050 & 0.038 & 1.456 & 1.448 & 1.443 & 1.434 & 1.429 & 1.346 & 0.469 \\
0.717 & 0.129 & 0.099 & 0.082 & 0.054 & 0.042 & 1.444 & 1.436 & 1.431 & 1.422 & 1.417 & 1.334 & 0.482 \\
0.813 & 0.116 & 0.088 & 0.072 & 0.046 & 0.035 & 1.467 & 1.459 & 1.454 & 1.445 & 1.440 & 1.358 & 0.456 \\
0.908 & 0.080 & 0.057 & 0.044 & 0.026 & 0.019 & 1.549 & 1.542 & 1.537 & 1.529 & 1.525 & 1.449 & 0.366 \\
\hline
\end{tabular}

}
\end{center}    
\end{table}

\clearpage

\begin{figure}
\centering
\begin{subfigure}[t]{0.3\textwidth}
\centering
\includegraphics[scale=0.2]{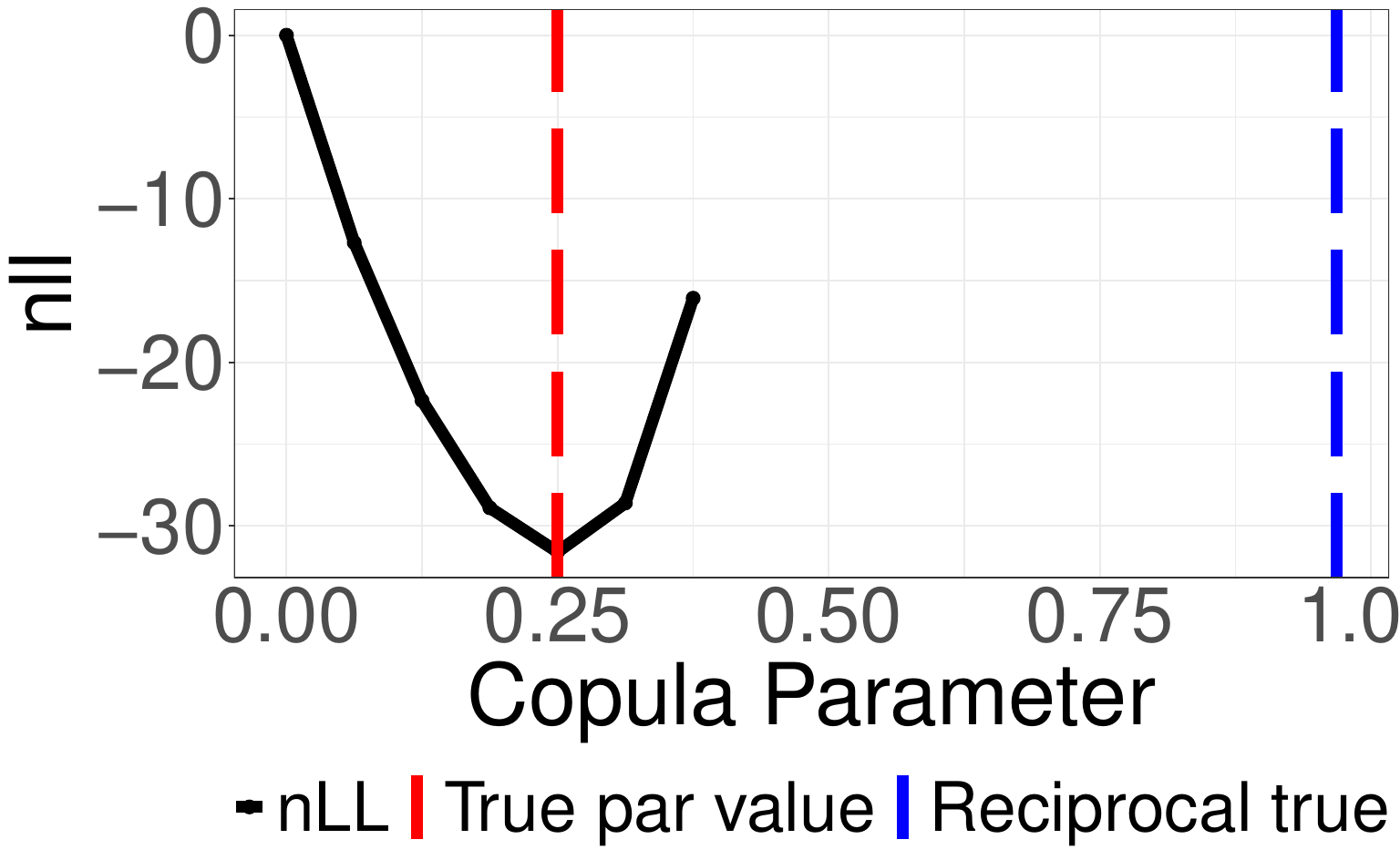}
\caption{NLL in relation to the MAG$(1)$ parameter value. The true parameter value is $\alpha = 0.25$.}
\end{subfigure}
~
\begin{subfigure}[t]{0.3\textwidth}
\centering
\includegraphics[scale=0.2]{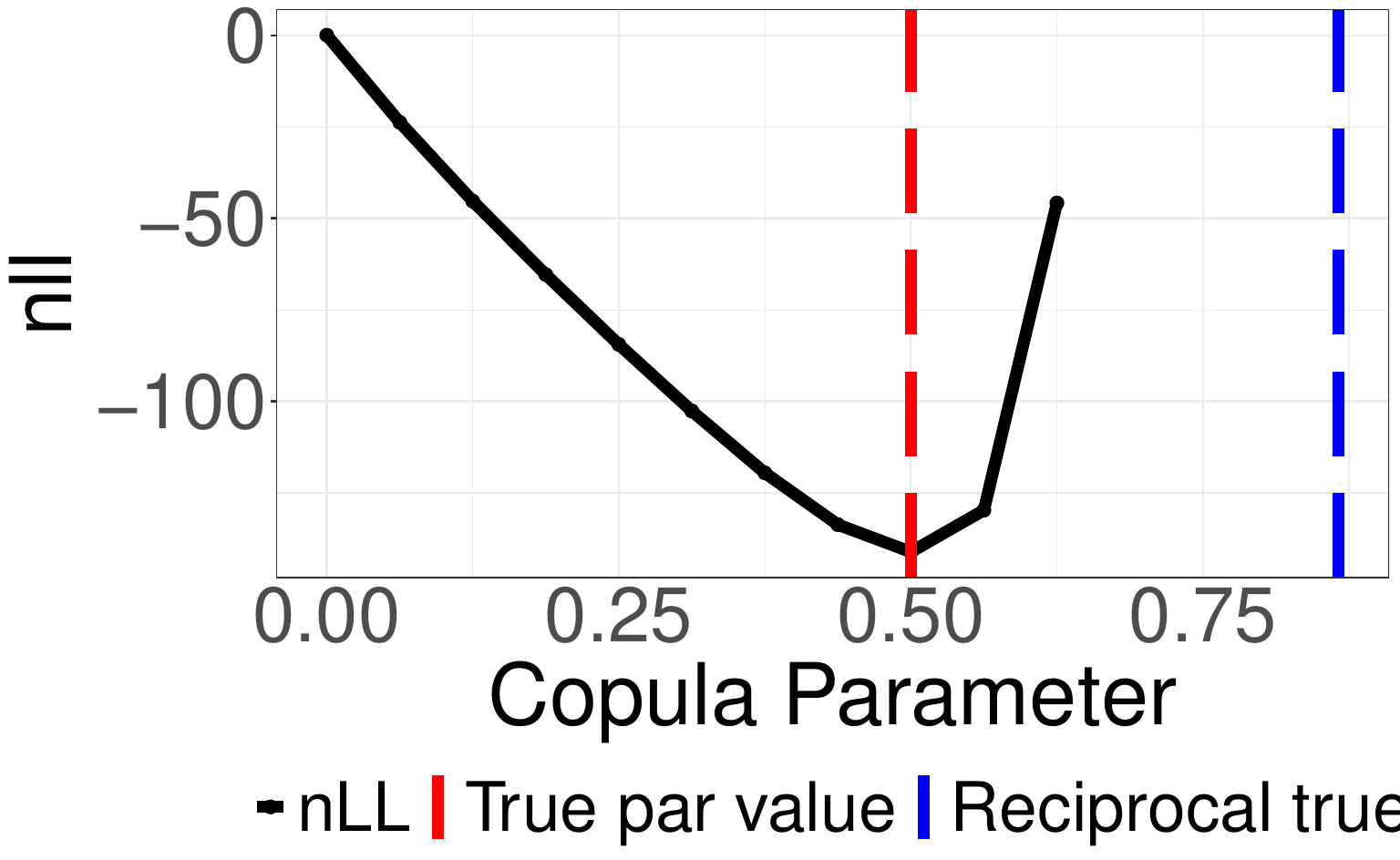}
\caption{NLL in relation to the MAG$(1)$ parameter value. The true parameter value is $\alpha = 0.5$.}
\end{subfigure}
~
\begin{subfigure}[t]{0.3\textwidth}
\centering
\includegraphics[scale=0.2]{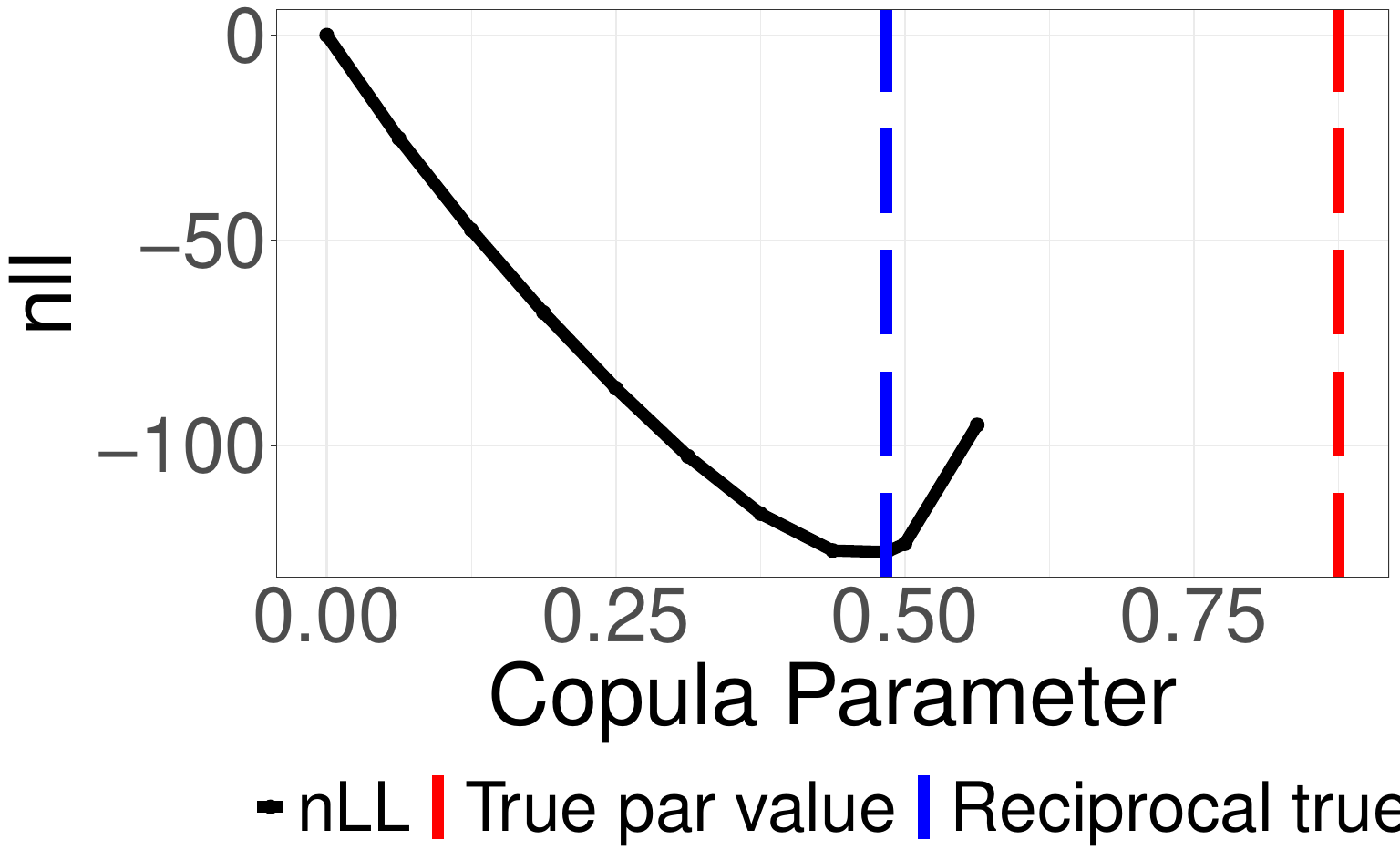}
\caption{NLL in relation to the MAG$(1)$ parameter value. The true parameter value is $\alpha = 0.875$.}
\end{subfigure}
\caption{The negative log-likelihood of a Gaussian-MAG$(1)$ model in dependence of the copula parameter. The black line is the NLL-value. The red line marks the true parameter value and the blue line the reciprocal true value, which is calculated as $\alpha^{\dagger} = \sqrt{1-\alpha^2}$.}
\label{Fig:nll_Gausian_MAG(1)}
\end{figure}

\begin{figure}
\centering
\begin{subfigure}[t]{0.3\textwidth}
\centering
\includegraphics[scale=0.2]{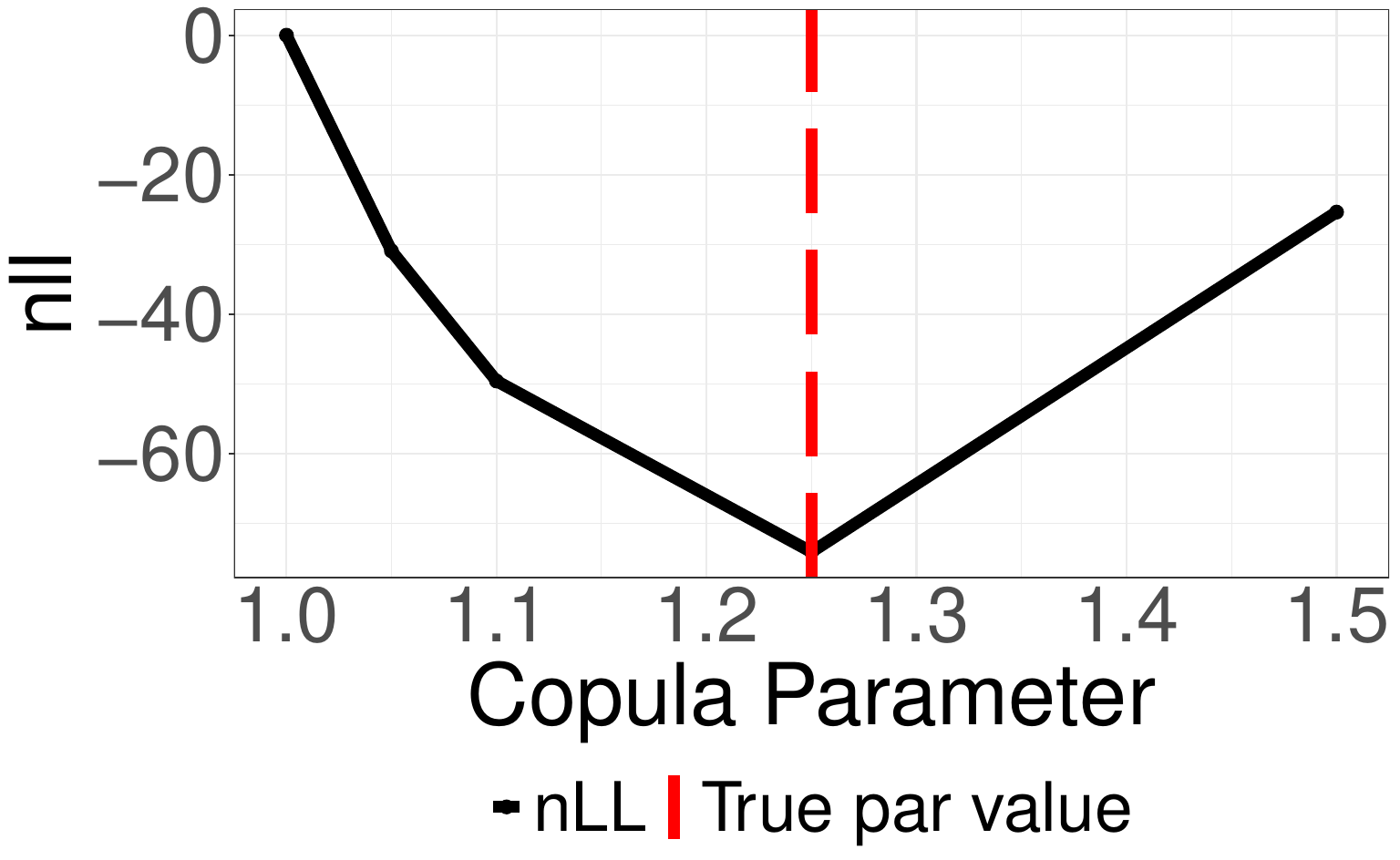}
\caption{NLL in relation to the MAG$(1)$ parameter value. The true parameter value is $\alpha = 1.25$.}
\end{subfigure}
~
\begin{subfigure}[t]{0.3\textwidth}
\centering
\includegraphics[scale=0.2]{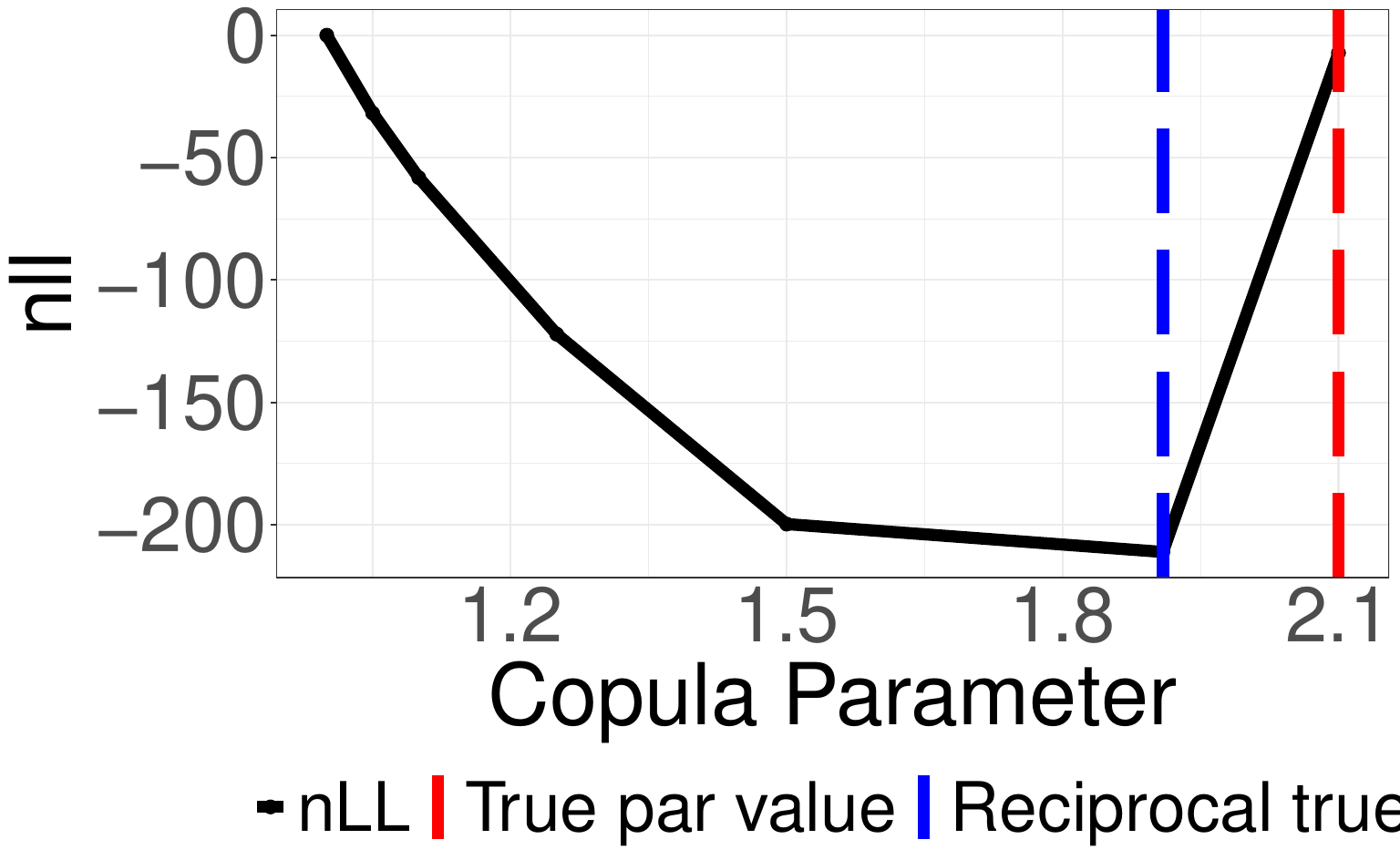}
\caption{NLL in relation to the MAG$(1)$ parameter value. The true parameter value is $\alpha = 2.1$.}
\end{subfigure}
~
\begin{subfigure}[t]{0.3\textwidth}
\centering
\includegraphics[scale=0.2]{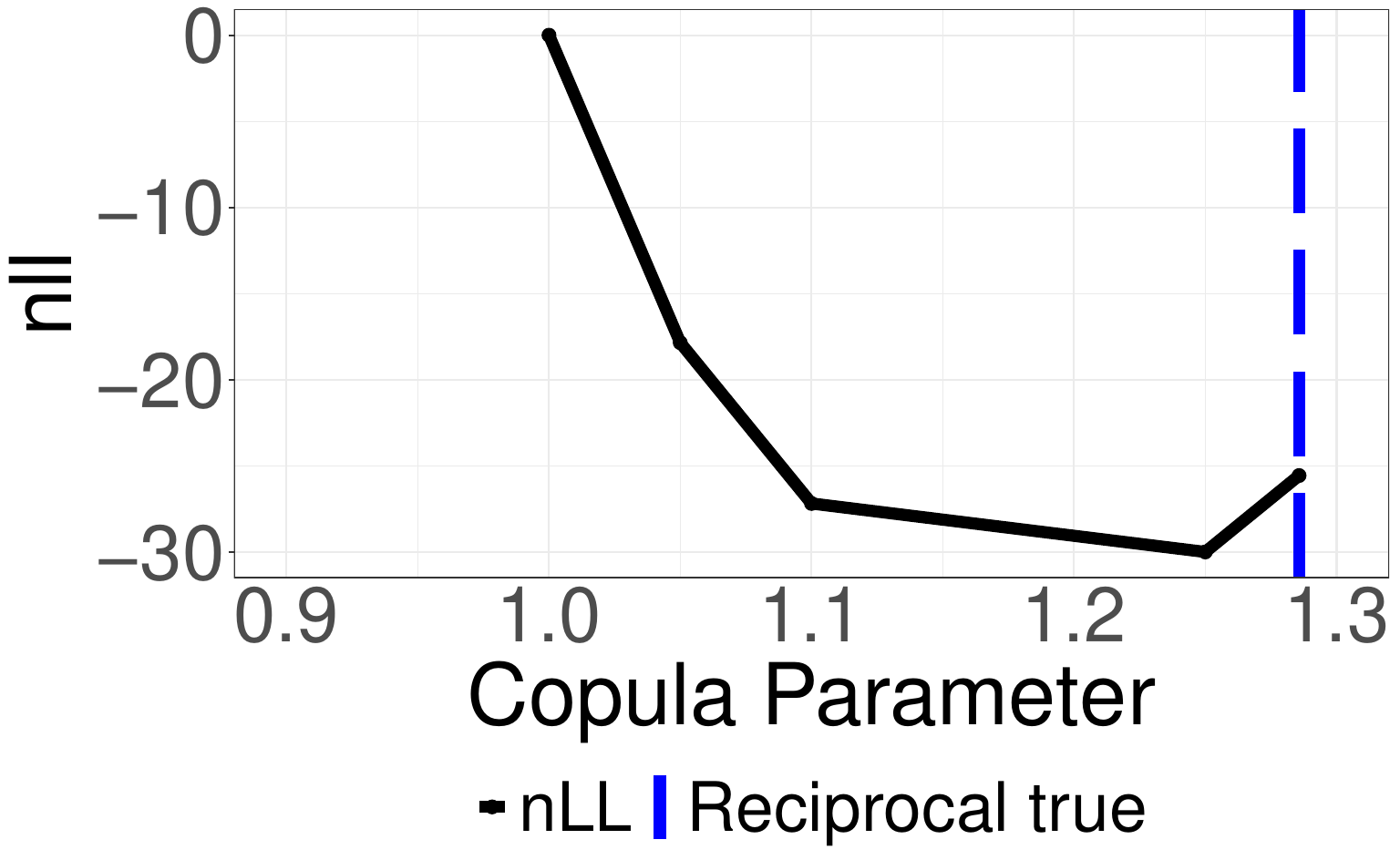}
\caption{NLL in relation to the MAG$(1)$ parameter value. The true parameter value is $\alpha = 4.5$.}
\end{subfigure}
\caption{The negative log-likelihood of a Gumbel-MAG$(1)$ model in dependence of the copula parameter. The black line is the NLL-value. The red line marks the true parameter value and the blue line the \SP{approximate reciprocal true value defined above.}}
\label{Fig:nll_Gumbel_MAG(1)}
\end{figure}

\begin{figure}
\centering
\begin{subfigure}[t]{0.3\textwidth}
\centering
\includegraphics[scale=0.2]{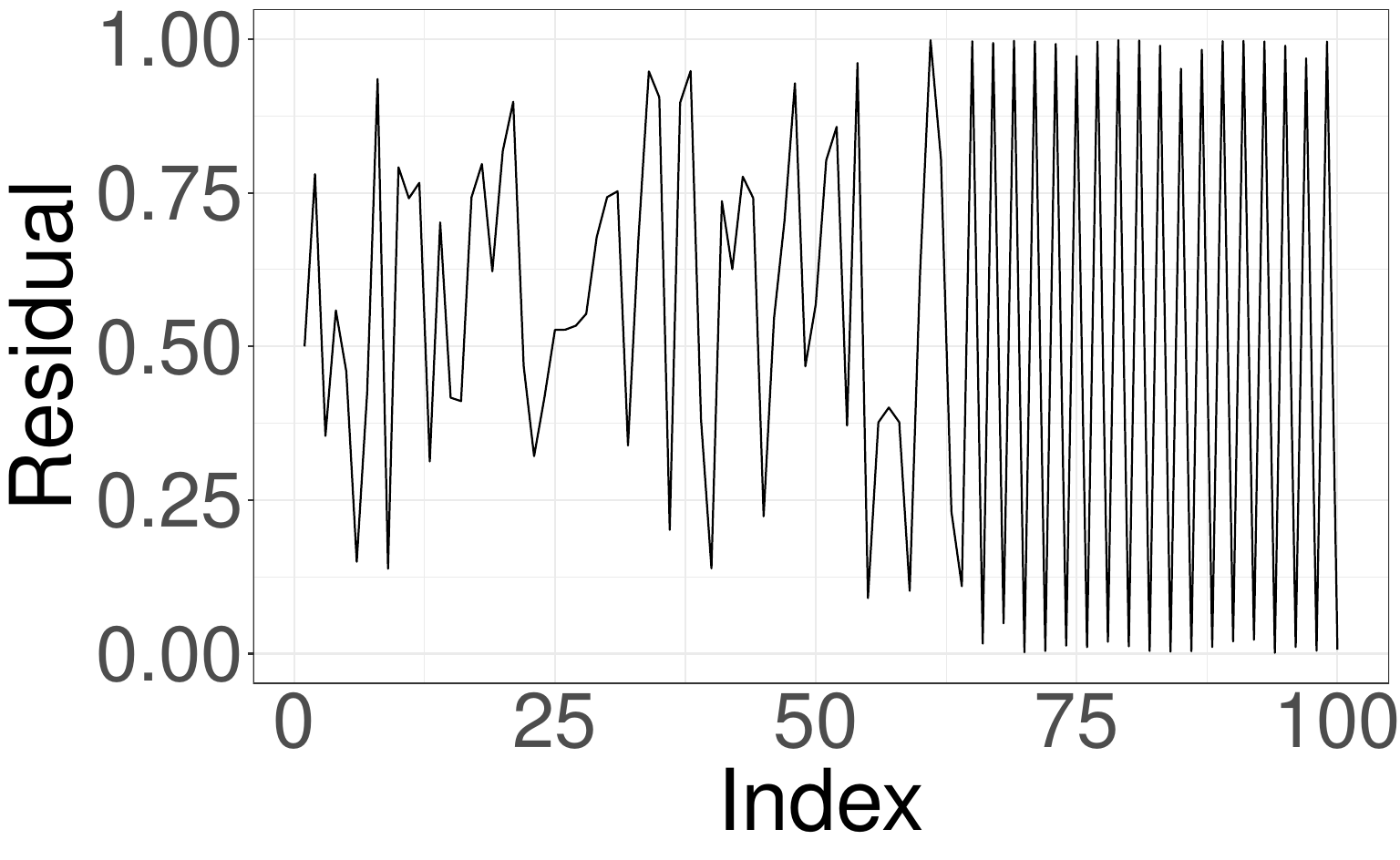}
\caption{rue parameter value.}
\end{subfigure}
~
\begin{subfigure}[t]{0.3\textwidth}
\centering
\includegraphics[scale=0.2]{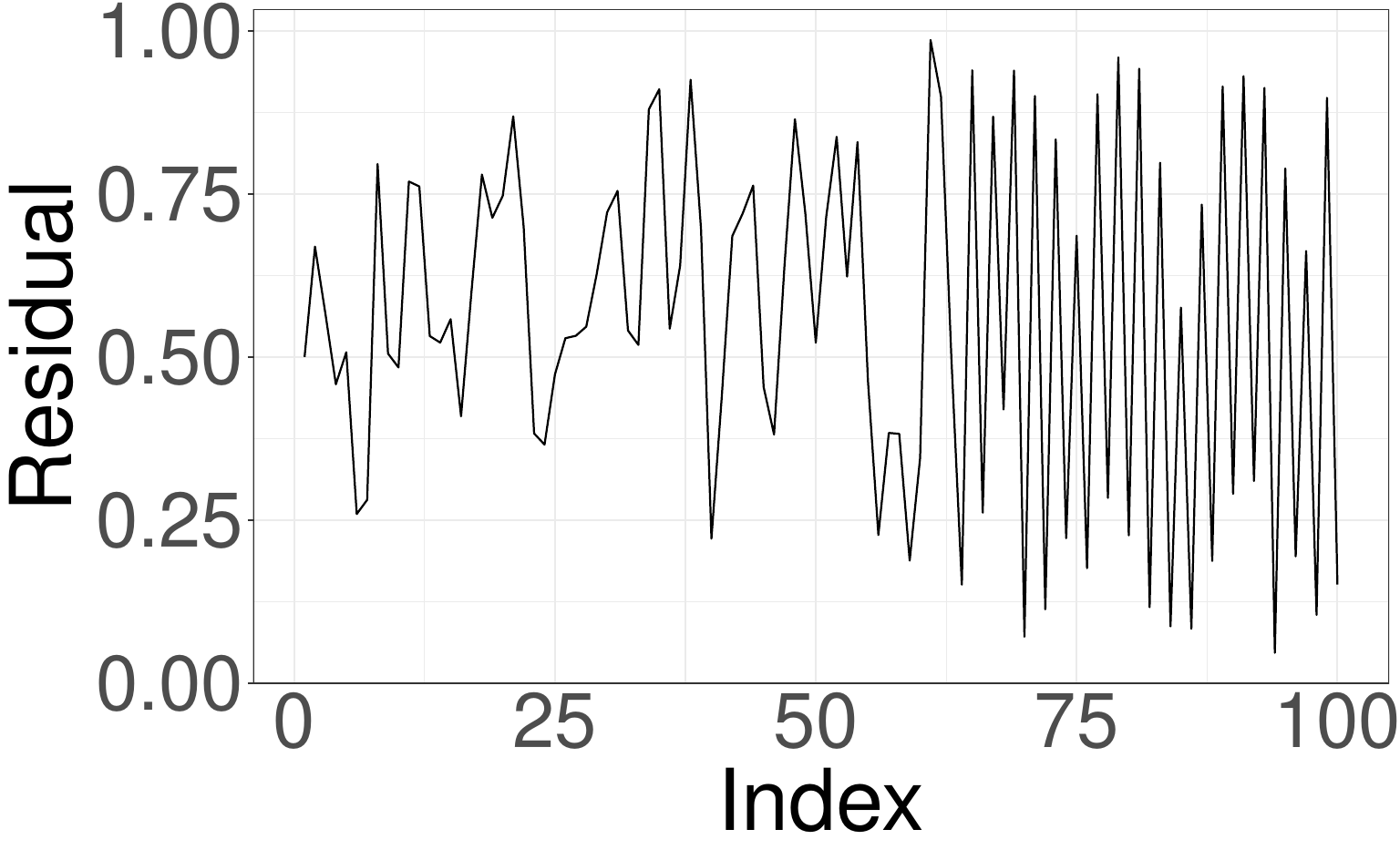}
\caption{Reciprocal parameter value.}
\end{subfigure}
\caption{Residuals of the Gaussian-MAG$(1)$ process. The left plot shows the residuals calculated with the true parameter value, which lies above the critical threshold. The right plot shows the residuals under the reciprocal of the true value.}
\label{Fig:MAG(1)_residuals}
\end{figure}

\clearpage

\section{Forecasting Real Data}
\label{Sec:Forecasting}

We use the \SP{model from Eq.~\ref{Eq:Model(p,q)_updating_Eq}} to obtain probabilistic forecasts for a univariate time series. Similar to the calculation of the likelihood, the construction of forecasts needs iterative calculations. In Algorithm~\ref{Algo:PARMA_forecast}, we present an algorithm that produces one-step ahead percentile forecasts for new data points. The iterative estimation of the latent processes $\{W_t\}$ and $\{\varepsilon_t\}$ is the same as for the likelihood calculation. Given these estimates the percentiles can be obtained by applying the conditional quantile function corresponding to the MAG-copula $K_{(p+1),\hdots,1}$.
The considered time series is German wind power production. The original, publicly available data set by \cite{Becker_2018, becker2017completion} comprises observations from different wind power farms in Germany with a 15min frequency over 16 years. We aggregate these time series to daily production for the whole of Germany, from which we consider the first $\sim 11$ years. We transform the time series using the function $f(x) = \log(x + 1)$, detrend the time series by estimating and subtracting a polynomial trend of second order and we subtract two cyclical components with period lengths $T_1 = 365$ and $T_2 = 1600$. %, which correspond to the dominant frequencies in the periodogram of the detrended series.
%
%In both forecasting studies
We take the first two years as train, the third as validation and the fourth as test set. The Gaussian-ARMA models, which serve as benchmarks, are selected by AIC on the training data using the \texttt{auto.arima} function from the \texttt{forecast} package in \texttt{R}, see \cite{hyndman2008automatic}. Selection is performed by grid-search and by stepwise selection. Additionally we use the Gaussian-ARMA$(1,1)$ and Gaussian-ARMA$(0,0)$ models as benchmarks. 
%Additionally we use the classic copula-based time series as benchmark. 
Additionally we use the Markov copula-based time series models as benchmark. 
For the selection of model order we again rely on automatic selection. We experiment with different copulas and chose the combination which performs best on the validation data. The same holds for the model from Eq.~\ref{Eq:Model(p,q)_updating_Eq}. The naming scheme of the models follows the naming by \cite{pappert2024moving}. First, the marginal specification is stated, ''kde'' for kernel-density estimation (from the \texttt{R}-package \texttt{kde1d}, \SP{see \cite{nager2025kde1d})} and ''n'' for the normal distribution. Then AR$(p)$ denotes the \SP{Markov} copula-based time series with model order $p$ and \SP{''CoARMA$(p,q)$'' is short for the model from Eq.~\ref{Eq:Model(p,q)_updating_Eq} with model order $p$ and $q$.} 
The letters in the parentheses indicate which copulas are used: 
''n'' stands for the Gaussian (or normal) copula, ''g'' for the Gumbel copula and ''t'' for the $t$-copula.
We note that the copula-based models are estimated in a two-step estimation scheme. First, the marginal distribution is estimated and second, the copula parameters are estimated on the probability integral transformed observations.\\
%
%
% Evaluation measures
%
As evaluation measures we use the continuous ranked probability score (CRPS), the predictive negative log-likelihood (NLL), the $5$/$95$-\% pinball score (PBS), the root mean squared error (RMSE), and the mean absolute error (MAE). The CRPS evaluates the whole probabilistic forecast. It estimates the integrated squared distance between the predicted CDF and the empirical CDF of the realized observation; compare \cite{gneiting2007strictly}. We approximate the CRPS by aggregating the pinball score values of all percentiles as is common in practice, \SP{compare e.g., Appendix D by \cite{hirsch2024online}.} Usually the CRPS is used for forecast evaluation only, we also report the in-sample CRPS to have a potential reference value.
The \SP{predictive} NLL also evaluates the whole predictive distribution. Note that the models are trained with respect to the NLL. The PBS evaluates quantile forecasts, the RMSE mean forecasts and the MAE median forecasts. For the copula models, the mean prediction is obtained by averaging all percentiles. The median prediction is just the 50\%-percentile. For the Gaussian-ARMA model mean and median prediction coincide. For all evaluation measures it holds that lower values mean better performance.\\
\begin{minipage}{\textwidth}
\begin{algorithm}[H]
\caption{Forecasting \SP{with the model from Eq.~\ref{Eq:Model(p,q)_updating_Eq}} based on \texttt{rvinecopulib}. \SP{\texttt{olddata} is the data that the model is trained on and new \texttt{newdata} the data on which the forecasting study in conducted. $F$ is a CDF estimated on the training data.}}\label{Algo:PARMA_forecast}
%\begin{algorithmic}
\SetKwInOut{Input}{input}
\SetKwInOut{Output}{output}
\Input{\texttt{olddata}, \texttt{newdata}, \texttt{model{\_}AR}, \texttt{model{\_}MA}, \SP{$p$, $q$} and $F$}% $m$, $F$}
\Output{matrix $\hat{U}$}
$e,x, W \gets \texttt{vector()}$ \\
%$p \gets \texttt{length(family{\_}AR)}$\\
%$q \gets \texttt{length(family{\_}MA)}$\\
$r \gets \max(p,q)$\\
$l, n \gets \texttt{length(olddata)}, \texttt{length(newdata)}$ \\
$x_{1:l} \gets \texttt{pobs(olddata)}$\\
$x_{(l+1):(l+n)} \gets F(\texttt{newdata})$\\
$e_{1:r}, W_{1:r} \gets (0.5,0.5)$\\
$\hat{U} = \texttt{matrix(ncol = n,nrow = 99)}$\\
\For{$t\gets (r+1)$ \KwTo $(l+n)$}{
	$e^{\dagger} \gets \texttt{inverse{\_}rosenblatt}(u = \left(W_{t-q}, e_{(t-q+1):(t-1)}, 0.5\right), \texttt{model{\_}MA})$\\
	$e_t \gets \texttt{rosenblatt}(x = (e^{\dagger}_{1:q},x_t), \texttt{model{\_}MA})_{(q+1)}$ \\
	$W^{\dagger} \gets \texttt{rosenblatt}(x = (W_{(t-p):(t-1)}, 0.5), \texttt{model{\_}AR})$\\
    $W_{t} \gets \texttt{inverse{\_}rosenblatt}((W^{\dagger}_{1:p}, e_t), \texttt{model{\_}AR})_{(p+1)}$\\
    %$d_t \gets \texttt{dvinecop}(u = (x_t, e_{(t-1):(t-q+1)}, W_{(t - q)}), \texttt{model{\_}MA})$
    }
\For{$t\gets (l+1)$ \KwTo $(n+l)$}{
	$\alpha \gets \texttt{1:99 / 100}$\\
	$U^{\dagger} \gets \texttt{rosenblatt}(x = (W_{(t-q)},e_{(t-q+1):(t-1)}, 0.5), \texttt{model{\_}MA})$\\
    $\hat{U}_{(t-l),(1:\SP{99)}} \gets \texttt{inverse{\_}rosenblatt}((U^{\dagger}_{(1:q)}, \alpha), \texttt{model{\_}MA})_{(q+1)}$
    }
    \Return $\hat{U}$
%\end{algorithmic}
\end{algorithm}
%}
\end{minipage}

\medskip

Now we discuss the forecasting study on German wind power production. The original aggregated time series comprises 16 years of daily observations of German wind power production.
We take the first two years as training data, the third year as validation and the fourth year as test set. The in-sample fitting results are given in Table~\ref{Table:wind_train_evaluation}. Among the Gaussian-ARMA models, the ARMA$(1,3)$ model fits best on the training data. The kde-AR$(3)$-(g,n,n) model is the best Markovian copula-based model, which is also the best fitting model on the training data. The \SP{kde-CoARMA$(3,1)$-(n,n,n)-(n)} model fits best among the \SP{CoARMA} models. On the validation data the \SP{kde-CoARMA$(3,1)$-(n,n,n)-(n)} model is best among the model class and among all models. The best models from each model class are then compared on the test data. We observe that the Markovian copula-based and the \SP{model from Eq.~\ref{Eq:Model(p,q)_updating_Eq}} outperform the Gaussian-ARMA model. 
The Markovian model of order 3 seems to predict wind power slightly better. 
The kde-CoARMA(1,1)-(g)-(t) with a Gumbel and t copula has comparable performance measures and is more parsimonious.
We conclude that estimating the stationary distribution with KDE enhances the performance. This is also confirmed by comparing the Gaussian NLL for the marginal distribution with the KDE NLL on training, validation and test data. On all instances the KDE NLL is smaller. Furthermore we conclude that wind power is dominated by linear relations, as is indicated by the normal copula performing well. However, the kde-CoARMA(1,1)-(g)-(t) model offers a parsimonious and non-linear alternative that performs well for out-of-sample.
This model also better matches the rank-based normal scores plot and lower/upper semi-correlations of consecutive values in the training set --- both of these diagnostics suggest upper tail dependence but not lower tail dependence.

% HJ
% normal scores plots and semi-correlations are diagnostics in Chapter 1 of my book.
% i1,i2= 1 730  (summary for first 730, similar resuls for other subsets)
%    Min.  1st Qu.   Median     Mean  3rd Qu.     Max. 
%-2.69592 -0.66385 -0.01870 -0.01318  0.62860  2.00112 
%  j1 j2      ncor      lcor      ucor  bvnsemic var1  var2
%   1  2 0.5786022 0.2270767 0.4781714 0.3357639    y ylag1

% bvnsemic = 0.336 is the upper/lower semi-correlation for Gaussian(ncor=0.579)
% There is stronger dependence in the joint upper tail, also seen in
% rank-based normal scores plots

% HJ notes: 16 years =~ 5840 days
% csv file has 4089 records

% g = gumbel? t = student 

%%%%%%%%%%%%%%%%%%%%%%%%%%%%%%%%%%%%%%%%%%%%%%%%%%%%%%
%%%%%%%%%%%%%%%%%%%%%%%%%%%%%%%%%%%%%%%%%%%%%%%%%%%%%%
%%%%%%%%%%%%%%%%%%%%%%%%%%%%%%%%%%%%%%%%%%%%%%%%%%%%%%

\clearpage
\begin{table}[!t]
%\begin{center}
\caption{\small \textbf{In sample} mean CRPS, NLL, RMSE and MAE values of fits on German wind power production. The time series contains 16 years of observations. The values displayed here show the fit on observations from the first two years (730 observations), which act as training data. The highlighted values are the ones which are best in the respective model class.}
\label{Table:wind_train_evaluation}
\centering
\resizebox{0.9\columnwidth}{!}{%
\begin{tabular}{p{5.7cm}|p{2cm}p{2cm}p{2cm}p{2cm}}
\hline\noalign{\smallskip}
Model & CRPS & NLL & RMSE & MAE
\\
\hline
\hline
Gaussian-ARMA Models & & & 
\\
\hline
Gaussian-ARMA(0,0) & $0.497$ & $923.9$ & $0.858$ & $0.713$
\\
Gaussian-ARMA(1,1) & $0.403$ & $775.9$ & $0.700$ & $0.575$
\\
Gaussian-ARMA(3,1) & $\textbf{0.401}$ & $773.6$ & $\textbf{0.698}$ & $\textbf{0.571}$
\\
Gaussian-ARMA(1,3) & $\textbf{0.401}$ & $\textbf{773.5}$ & $\textbf{0.698}$ & $\textbf{0.571}$
\\
\hline
\hline
Copula Markov models
\\
\hline
kde-AR$(3)$-(g,n,n) & $\textbf{0.399}$ & $\textbf{751.9}$ & $\textbf{0.697}$ & $\textbf{0.567}$
\\
kde-AR$(3)$-(t,n,n) & $0.401$ & $756.7$ & $0.700$ & $0.570$
\\
\hline
\hline
CoARMA-models
\\
\hline
kde-CoARMA$(3,1)$-(n,n,n)-(n) & $\textbf{0.401}$ & $\textbf{757.1}$ & $\textbf{0.699}$ & $\textbf{0.570}$
\\
kde-CoARMA$(3,1)$-(g,n,n)-(n) & $0.402$ & $759.4$ & $0.700$ & $0.572$
\\
kde-CoARMA$(1,1)$-(g)-(t) & $0.402$ & $758.1$ & $0.700$ & $\textbf{0.570}$
\\
kde-CoARMA$(2,2)$-(g,n)-(n,t) & $0.441$ & $855.6$ & $0.764$ & $0.634$
\\
\hline
\end{tabular}
}
%\end{center}
\end{table}

\begin{table}[!t]
%\begin{center}
\caption{\small \textbf{Validation} mean CRPS, NLL, Pinball 5\% \& 95\%, RMSE and MAE values of forecasts of German wind power. The evaluation is performed on one year of data (365 observations). The highlighted values are the ones which are best in the respective model class.}
\label{Table:Wind_val_evaluation}
\centering
\resizebox{0.9\columnwidth}{!}{%
\begin{tabular}{p{5.7cm}|p{1.5cm}p{1.5cm}p{1.5cm}p{1.5cm}p{1.5cm}p{1.5cm}}
\hline\noalign{\smallskip}
Model & CRPS & NLL & PBS-5\% & PBS-95\% & RMSE & MAE
\\
\hline
\hline
Gaussian-ARMA Models & & & & & &
\\
\hline
Gaussian-ARMA(0,0) & $0.535$ & $489.5$ & $0.191$ & $0.159$ & $0.920$ & $0.762$
\\
Gaussian-ARMA(1,1) & $0.409$ & $393.4$ & $\textbf{0.151}$ & $0.131$ & $0.711$ & $0.578$
\\
Gaussian-ARMA(3,1) & $\textbf{0.406}$ & $\textbf{391.3}$ & $0.152$ & $\textbf{0.129}$ & $\textbf{0.707}$ & $\textbf{0.574}$
\\
Gaussian-ARMA(1,3) & $0.407$ & $391.5$ & $0.152$ & $\textbf{0.129}$ & $\textbf{0.707}$ & $\textbf{0.574}$
\\
\hline
\hline
Copula Markov models & & & & & &
\\
\hline
kde-AR$(3)$-(g,n,n) & $0.408$ & $388.9$ & $\textbf{0.153}$ & $0.132$ & $0.715$ & $0.574$
\\
kde-AR$(3)$-(t,n,n) & $\textbf{0.406}$ & $\textbf{385.2}$ & $0.155$ & $\textbf{0.128}$ & $\textbf{0.709}$ & $\textbf{0.570}$
\\
\hline
\hline
CoARMA-models & & & & & &
\\
\hline
kde-CoARMA$(3,1)$-(n,n,n)-(n) & $\textbf{0.406}$ & $\textbf{384.8}$ & $\textbf{0.152}$ & $\textbf{0.128}$ & $\textbf{0.709}$ & $0.571$
\\
kde-CoARMA$(3,1)$-(g,n,n)-(n) & $0.407$ & $385.7$ & $0.154$ & $\textbf{0.128}$ & $0.711$ & $0.573$
\\
kde-CoARMA$(1,1)$-(g)-(t) & $0.407$ & $386.7$ & $0.158$ & $0.130$ & $0.713$ & $\textbf{0.571}$
\\
kde-CoARMA$(2,2)$-(g,n)-(n,t) & $0.463$ & $443.8$ & $0.173$ & $0.139$ & $0.805$ & $0.658$
\\
\hline
\end{tabular}
}
%\end{center}
\end{table}

\begin{table}[!t]
%\begin{center}
\caption{\small \textbf{Test} mean CRPS, NLL, Pinball 5\% \& 95\%, RMSE and MAE values of forecasts of German wind power production. The evaluation is performed on one year of test data (365 observations). The models shown here are the ones which performed best in their respective model class on the validation data. The highlighted values are the lowest ones in the respective column.}
\label{Table:Wind_test_evaluation}
\centering
\resizebox{0.9\columnwidth}{!}{%
\begin{tabular}{p{5.7cm}|p{1.5cm}p{1.5cm}p{1.5cm}p{1.5cm}p{1.5cm}p{1.5cm}}
\hline\noalign{\smallskip}
Model & CRPS & NLL & PBS-5\% & PBS-95\% & RMSE & MAE
\\
\hline
\hline
Gaussian-ARMA(3,1) & $0.391$ & $379.7$ & $0.140$ & $0.140$ & $0.684$ & $0.548$
\\
kde-AR$(3)$-(t,n,n) & $\textbf{0.389}$ & $\textbf{372.2}$ & $0.139$ & $\textbf{0.136}$ & $\textbf{0.682}$ & $0.546$
\\
kde-CoARMA$(3,1)$-(n,n,n)-(n) & $\textbf{0.389}$ & $372.7$ & $\textbf{0.138}$ & $0.138$ & $\textbf{0.682}$ & $\textbf{0.545}$
\\
kde-CoARMA$(1,1)$-(g)-(t) & $0.392$ & $375.2$ & $0.141$ & $\textbf{0.136}$ & $0.688$ & $0.553$
\\
\hline
\end{tabular}
}
%\end{center}
\end{table}
\clearpage

\section{Discussion}
\label{Sec:Discussion}

We analyze the copula ARMA extension, Eq.~\ref{Eq:Model(p,q)_updating_Eq}, \SP{proposed in Section 3.14 in \cite{joe2014dependence}.} The main feature of this process is that it is a copula generalization of the ARMA model with uniform stationary distribution.
% Relation to Gaussian-ARMA
We show that the model nests some Gaussian-ARMA processes in Proposition \ref{prop:GaussianARMA}. It is revealed that a $\Phi$-quantile transformed \SP{process from Eq.~\ref{Eq:Model(p,q)_updating_Eq}} of order $p$ and $q$, with AR and MAG copulas being the Gaussian copula, is a subset of a Gaussian-ARMA$(p,q+p-1)$ process. The additional MA-parameters are all completely determined by the copula parameters. This result is surprising since the \SP{process involves a Markov copula of order $p$ and a $q$-dependent copula.} The derivation reveals that there are additional terms. The additional terms are due to the structure of the process. For $p=1$, there are no additional terms though. This finding brings up the question whether there is a process composed of MA-type and AR-type parts \SP{that yields a Gaussian ARMA$(p,q)$ as a special case for all $(p,q) \in \mathbb{N}^2$.} The MAGMAR-copula model by \cite{pappert2024moving} can generalize ARMA models with any order, however, in this approach an additional transformation is necessary to ensure a uniform stationary distribution. This is a drawback as theoretical properties are not as easy to derive and the computational demand is increased. \\
We also show that, with a suitable choice of AR and MAG copulas, the Gaussian-GARCH$(1,1)$ process can be recovered, compare Proposition~\ref{prop:GARCH(1,1)}. In principle this enables modeling of time series with GARCH$(1,1)$-dynamics and arbitrary marginal distributions and complements the research on GARCH-mimicking copulas by \cite{dias2024garch}. The estimation of such processes, however, is not straightforward \SP{and could be subject to future research.}
We proceed to investigate the copula-MAG$(1)$ process in detail. This process is a basic building block of the model under investigation and also of the MAGMAR-copula model by \cite{pappert2024moving}. Simulations, numerical approximations and exemplary calculations on the Fr{\'e}chet-MAG$(1)$ process reveal that the dependence of the joint distribution of consecutive observations from a copula MAG$(1)$ process is limited. We regularly find that the ail dependence coefficients associated to this joint distribution seem to be bounded by at least $\frac{1}{4}$. Numerical approximation of the tail order indicates that the tail dependence coefficients may also vanish for ''standard''-copulas (Gumbel, Clayton and $t$). For the Fr{\'e}chet-MAG$(1)$ process, the upper bound of $\frac{1}{4}$ can also be found analytically.\\
%
% Identifiability
%
Furthermore, we find that the Gaussian-MAG$(1)$ process, similar to the classic MA$(1)$ process, has two equivalent representations. The two representations can be obtained by permuting the innovations. This gives rise to potential \SP{near non-identifiability problems.} We argue that restricting the parameter space \SP{to avoid stronger dependence parameters is not problematic.} This bound of the parameter space coincides with the bound that is needed such that the estimation is consistent (cf.~Proposition 7 by \cite{pappert2024moving}). It is not obvious \SP{whether} other MAG$(1)$ processes have two equivalent representations.
%
% PARMA(1,1)
%
We proceed to examine the dependence of consecutive observations of the \SP{model from Eq.~\ref{Eq:Model(p,q)_updating_Eq}} with model orders $p=q=1$ and discuss distributional, dependence and estimation properties of the resulting process. A crucial condition for consistency is the stationarity and ergodicity of the estimated latent processes. This condition is on a high-level. Currently it is not clear how to translate this condition to properties of the AR and MAG copula. This question offers potential for further investigations.\\
For the general model with arbitrary model orders, we provide an algorithm, which calculates the likelihood iteratively. Additionally we provide an algorithm that performs one-step ahead probabilistic forecasting.\SP{(The implemented \texttt{R}-code for estimation, simulation and forecasting is available upon request.)} Using these algorithms we fit models German wind power production and conduct a forecasting study. We observe that estimating the stationary distribution with kernel-density estimation enhances the models performances in comparison with models with normal stationary distributions. Linear temporal dependencies are revealed to dominate in terms of forecasting. However, the model from Eq.~\ref{Eq:Model(p,q)_updating_Eq} with the Gumbel copula as AR-copula and the $t$-copula as MAG-copula offers a well-performing and parsimonious alternative. In future research it would be interesting to apply the model to temporal dependencies, which are dominated by truly non-linear dynamics.

\section*{Funding Statement}
\noindent
\SP{The first author gratefully acknowledges (partial) funding in the course of TRR 391 \textit{Spatio-temporal Statistics for the Transition of Energy and Transport} (520388526) by the Deutsche Forschungsgemeinschaft (DFG, German Research Foundation). The second author gratefully acknowledges funding in the course of Discovery Grant GR010293 by the Natural Sciences and Engineering Research Council of Canada (NSERC).}

\section*{Conflict of Interest Statement}
The authors declare that there is no conflict of interest.

\section*{Data Availability Statement}
The data used for the analysis is aggregated from a publicly available data source, which is cited in the manuscript.

\section*{Associated Key Words}
Copula-based time series, Gaussian ARMA, Identifiability, Markov dependence, Probabilistic Forecasting, Tail dependence

\section*{MSC (Mathematics Subject Classification)}
62H05, 62M09, 62M10, 62P12, 60G10

\newpage
\bibliography{bibi_14102026.bib}

\newpage
\appendix

\section{Proofs}
\subsection{Proof of Proposition \ref{prop:GaussianARMA}}
\label{App:Proof:GaussianARMA}

\begin{proof}
We go step-by-step.
\begin{itemize}
\item[a)] We use the expression from \cite{mcneil2022time} for the last element of the Rosenblatt transform associated to a stationary D-vine copula with Gaussian pair copulas. Using their result, we may write
\begin{align}
\SP{g^{-1}}(x;u_1,\hdots,u_p) = \Phi\left( \frac{\Phi^{-1}(x) - \sum_{j = 1}^p\phi^{(p)}_j \Phi^{-1}(u_j)}{\sigma_{\alpha}} \right)
\end{align}
where $\sigma_{\alpha}^2 = \prod_{j=1}^p(1 - \alpha^2_j)$ and $\phi^{(p)}_j, j \in \{1,\hdots,p\}$ are given as in part a) of Proposition~\ref{prop:GaussianARMA}. It immediately follows that
\begin{align*}
\Phi^{-1}(U_t) = \sum_{j = 1}^p \phi^{(p)}_j \Phi^{-1}(U_{t-j}) + \sigma_{\alpha} \varepsilon_t,
\end{align*}
which proves the claim. 
%Again we note that this result is well-known in the literature. 
The result follows directly from \cite{mcneil2022time}. Compare also \cite{smith2010modeling}. The proof is given here for the sake of completeness.

\item[b)]
%\SPP
{The mapping $h$, which is used to construct the MAG-process, is the conditional quantile function corresponding to a copula $K_{1,\hdots,(q+1)}$. The copula $K_{1,\hdots,(q+1)}$ is not a stationary D-vine, but a MAG D-vine, i.e., a vine copula with some independent components. This means that $K_{1,\hdots,(q+1)}$ has independence pair-copulas
except for $K_{12}, K_{13;2} ,\ldots, K_{1,q+1|1:q}$.}

The last element of the Rosenblatt transform associated to such a vine is given as
\begin{align}
\SP{h}^{-1}(x;u_1,\hdots,u_q) = h_{\beta_q}(\,\cdot\,;u_q) \circ \hdots \circ h_{\beta_1}(x;u_1),
\end{align}
\SP{where $h_{\beta_i}$ are conditional quantile functions corresponding to the bivariate copulas, which the MAG D-vine is comprised of. This representation follows from the independent margins.}
For Gaussian pair copulas, we observe that the following structure emerges,
\begin{align}
\SP{h}^{-1}(x;u_1,\hdots,u_q) = \Phi\left( \sigma_{\beta}^{-1} \Phi^{-1}(x) - \sum_{i = 0}^{q-1} \left[\prod_{j = 0}^i (1-\beta_{q-j}^2)^{-1/2} \right] \beta_{q-i} \Phi^{-1}(u_{i+1}) \right).
\end{align}
Then we see
\begin{align}
    U_t =& \SP{h}(\varepsilon_t;\varepsilon_{t-1},\hdots,\varepsilon_{t-q}) \\
    \iff \varepsilon_{t} =& \SP{h}^{-1}(U_t;\varepsilon_{t-1},\hdots,\varepsilon_{t-q})
    \\
    \iff \sigma_{\beta} \Phi^{-1}(\varepsilon_t) =& \Phi^{-1}(U_t) - \sigma_{\beta}\sum_{i = 0}^{q-1} \left[\prod_{j = 0}^i (1-\beta_{q-j}^2)^{-1/2} \right] \beta_{q-i}\Phi^{-1}(\varepsilon_{t-q + i})
    \\
    \iff Y_t =& \eta_t + \sum_{i = 0}^{q-1} \left[\prod_{j = 0}^i (1-\beta_{q-j}^2)^{-1/2} \right] \beta_{q-i} \eta_{t-q+i}
    \\
    =& \eta_t + \sum_{j = 1}^q \theta_j^{(q)} \eta_{t - j}
\end{align}
where $\theta_j^{(q)} = \sigma_{\beta}^{-1}\left[ \prod_{k = 0}^{j-1} (1 - \beta_k^2)^{1/2} \right] \beta_j$ and $\eta_t = \sigma_{\beta} \Phi^{-1}(\varepsilon_t)$. \SP{Alternatively, this result may also be obtained be realizing that $(\beta_1, \hdots, \beta_q)$ parametrize a conditional correlation matrix of the underlying Gaussian distribution. Converting the partial correlations to unconditional correlations yields the same result.}

\item[c)] Following the procedure from part b), we have
\begin{align}
Y_t =& \Phi^{-1}(U_t)
\\
=& \eta_t + \sum_{i = 1}^{q-1} \theta_{i}^{(q)} \eta_{t-i} + \theta_{q}^{(q)} \sigma_{\beta} X_{t-q}.
\end{align}
Where it is understood that the sum $\sum_{i=1}^{q-1} \theta_{i}^{(q)} \eta_{t-i}$ is zero for $q = 1$. Now we see that $X_{t-q} = \sum_{j = 1}^p \phi^{(p)}_j X_{t-q-j} + \sigma_{\alpha} \varepsilon_{t-q}$ and for each $(X_{t-q-1},\hdots,X_{t-q-p})$ it holds that
\begin{align}
\theta_q^{(q)} \sigma_{\beta} X_{t-q-j} = Y_{t-j} -  \eta_{t-j} - \sum_{i = 1}^{q-1} \theta_i^{(q)} \eta_{t-j-i}
\end{align}
Consequently,
\begin{align}
Y_t =& \eta_t + \sum_{i = 1}^{q-1} \theta_i^{(q)} \eta_{t-i} + \sum_{j = 1}^p \theta_q^{(q)} \sigma_{\beta} \phi_j^{(p)} X_{t-q-j} + \theta_q^{(q)} \sigma_{\beta} \sigma_{\alpha} \varepsilon_{t-q}
\\
=& \eta_t + \sum_{i = 1}^{q-1} \theta_i^{(q)} \eta_{t-i} + \theta_q^{(q)} \sigma_{\alpha} \eta_{t-q}
\\
&+ \sum_{j = 1}^p \phi_j^{(p)} \Bigl[ Y_{t-j} -  \eta_{t-j} - \sum_{i = 1}^{q-1} \theta_i^{(q)} \eta_{t-j-i}\Bigr]
\\
=& \eta_t + \sum_{j = 1}^p \phi^{(p)}_j Y_{t-j} + \sum_{i = 1}^{q-1} \theta_i^{(q)} \eta_{t-i} + \theta_q^{(q)} \sigma_{\alpha} \eta_{t-q} - \sum_{j = 1}^p \phi^{(p)}_j \eta_{t-j} - \sum_{j = 1}^p\sum_{i = 1}^{q-1} \phi_j^{(p)} \theta_i^{(q)} \eta_{t-j-i}
\\
=& \eta_t + \sum_{j = 1}^p \phi^{(p)}_j Y_{t-j} + \sum_{j = 1}^{p + q - 1} \Psi_j^{(p,q)} \eta_{t-j}
\end{align}
where
\begin{align}
\Psi_j^{(p,q)} =& \tilde{\theta}_j^{(q)} - \tilde{\phi}^{(p)}_j - \zeta_j^{(p,q)}
\\
\tilde{\theta}_j^{(q)} =& \theta_j^{(q)} \times \mathds{1}\{j \leq q \land q > 1\} + \theta_q \sigma_{\alpha} \times \mathds{1}\{j = q\}
\\
\tilde{\phi}_j^{(p)} =& \phi_j^{(p)} \times \mathds{1}\{j \leq p\}
\\
\zeta_j^{(p,q)} =& \left(\sum_{(r,s) \in S_j^{(p,q)}} \phi^{(p)}_r \theta_s^{(q)}\right) \times \mathds{1}\{q>1\}
\\
S_j^{(p,q)} =& \{(r,s) \in \mathbb{N}^2: r \leq p, s \leq q, r+s = j\}
\end{align}
\end{itemize}
\end{proof}

\subsection{Proof of Proposition~\ref{prop:ARCH(1)}}
\label{App:Proof:ARCH(1)}

We derive the copula directly from the ARCH$(1)$ process,
\begin{align}
Y_t =& \sigma_t \eta_t, \quad \eta_t = \Phi^{-1}(\varepsilon_t) \stackrel{iid}{\sim} N(0,1)
\\
\sigma_t =& \sqrt{\alpha_0 + \alpha_1 Y_{t-1}^2}.
\end{align}
The unconditional distribution function of $Y_t$ is denoted by $\Psi$ and the corresponding PIT is denoted by $U_t = \Psi(Y_t)$. \SP{(The explicit form of $\Psi$ (or its derivative, the density) is derived in e.g. \cite{dias2024garch}.)}  Then we can write
\begin{align}
Y_t =& \sqrt{\alpha_0 + \alpha_1 Y_{t-1}^2} \eta_t
\\
U_t =& \Psi\left[ \sqrt{\alpha_0 + \alpha_1 \Psi^{-1}(U_{t-1})^2} \Phi^{-1}(\varepsilon_t) \right]
\\
=:& g(\varepsilon_t, U_{t-1})
\end{align}
Now we identify the mapping $g$ as the conditional quantile function corresponding to a copula $C_{21}$, which is evaluated at $\varepsilon_t$, conditioned on $U_{t-1}$, i.e., $U_t = g(\varepsilon_t, U_{t-1}) = C_{2|1}^{-1}(\varepsilon_t | U_{t-1})$. It follows that $C_{2|1}^{-1}$ is given as
\SP{
\begin{align}
C_{2|1}^{-1}(s|u_1) = \Psi\left[ \sqrt{\alpha_0 + \alpha_1 \Psi^{-1}(u_1)^2} \Phi^{-1}(s) \right], \quad 0 < s < 1, \ 0 < u_1 < 1.
\end{align}
}
Then we can immediately see that the conditional CDF corresponding to $C_{12}$ is given as
\begin{align}
C_{2|1}(u_2|u_1) = \Phi\left[ \frac{\Psi^{-1}(u_2)}{\sqrt{\alpha_0 + \alpha_1 \Psi^{-1}(u_1)^2}} \right]
\end{align}
and the expression for the $C_{12}$ is obtained by integrating over the second argument. Similarly the copula density $c_{12}$ can be derived by taking the derivative of $C_{2|1}$ with respect to its first argument. Then the expression given in Eq.~9 of \cite{dias2024garch} can be recovered.

\subsection{Proof of Proposition~\ref{prop:t-copula_ARCH(1)}}
\label{App:Proof:t-copula_ARCH(1)}

\SP{
Stationarity of the process follows directly since $0 \leq \alpha = 1 / (\nu +1) < 1$ for $\nu >0$. Now we construct $\{\tilde{Y}_t\}_{t\in\mathbb{Z}}$ as
\begin{align}
\tilde{Y}_t = t_{\nu}^{-1} [ C_{2|1}^{-1}(\varepsilon_t | t_{\nu}(\tilde{Y}_{t-1})],
\end{align}
where $C_{21}$ is the $t$-copula with $\rho = 0$ and $\nu > 0$. Using the explicit expression for the conditional quantile function associated to the $t$-copula (compare e.g. \cite{smith2010modeling}), we see that
\begin{align}
\tilde{Y}_t = \sqrt{(\nu + \tilde{Y}_{t-1}^2) / (\nu + 1)} \eta_t,
\end{align}
which is the defining equation of the ARCH process in Prop.~\ref{prop:t-copula_ARCH(1)}, hence $\{Y_t\} = \{\tilde{Y}_t\}$. Since $\tilde{Y}_t$ is obtained by quantile-transforming with respect to $t_{\nu}$ we conclude that the stationary distribution is $t_{\nu}$.
}

\subsection{Proof of Proposition~\ref{prop:GARCH(1,1)}}
\label{App:Proof:GARCH(1,1)}

We follow the same strategy as in the derivation for the ARCH-copula in Sect.~\ref{App:Proof:ARCH(1)}. The main difference here is that the MAG and the AR part need to be identified. Let $\Psi$ be the unconditional CDF corresponding to a Gaussian-GARCH$(1,1)$ process and let $\{X_t\}_{t\in\mathbb{Z}}$ follow the modified-ARCH$(1)$ model given by the following updating equation
\begin{align}
X_t = \sqrt{(\alpha_0 + X_{t-1}^2)(\alpha_1 \eta_{t-1}^2 + \beta)}.
\end{align}
Using contraction arguments by \cite{douc2014nonlinear} we see that $\{X_t\}_{t\in\mathbb{Z}}$ is stationary if $\frac{1}{2}\mathbb{E}[\log(\alpha_1 \eta_0^2 + \beta)] < 0$, which is \SP{a $\frac{1}{2}$-scaled version of the GARCH$(1,1)$ stationarity condition and hence} automatically fulfilled if the assumed GARCH$(1,1)$ is stationary.
Now let $\Gamma$ be the unconditional CDF of $\{X_t\}_{t\in\mathbb{Z}}$. The GARCH$(1,1)$ model is given as
\begin{align}
Y_t =& \sigma_t \eta_t
\\
\sigma_t =& \sqrt{\alpha_0 + \alpha_1 Y_{t-1}^2 + \beta_1 \sigma_{t-1}^2}
\end{align}
Now we define \SP{for all $t\in\mathbb{Z}$
\begin{align}
X_{t-1} :=& \sqrt{\alpha_1 Y_{t-1}^2 + \beta_1 \sigma_{t-1}^2}
%\\
%=& \sqrt{\alpha_1 Y_{t-1}^2 + \beta_1 \frac{Y_{t-1}^2}{\eta_{t-1}^2}}
= \sqrt{\alpha_1 Y_{t-1}^2 + \beta_1 \frac{Y_{t-1}^2}{\eta_{t-1}^2}}
\label{Eq:GARCH(1,1)_Xt}
\end{align}
}
and hence
\SP{
\begin{align}
Y_t =& \sqrt{\alpha_0 + X_{t-1}^2} \eta_t \quad \text{and}
%\\
%Y_{t-1} =& \sqrt{\alpha_0 + X_{t-2}^2} \eta_{t-1}
Y_{t-1} & \sqrt{\alpha_0 + X_{t-2}^2} \eta_{t-1}
\end{align}
}
Now we can substitute for $X_t$ in Eq.~\ref{Eq:GARCH(1,1)_Xt} and obtain
\SP{
\begin{align}
X_t =& \sqrt{\alpha_1 (\alpha_0 + X_{t-1}^2) \eta_{t}^2 + \beta_1 (\alpha_0 + X_{t-1}^2)}
%\\
%=& \sqrt{(\alpha_0 + X_{t-1}^2)(\alpha_1 \eta_{t}^2 + \beta_1)}
= \sqrt{(\alpha_0 + X_{t-1}^2)(\alpha_1 \eta_{t}^2 + \beta_1)}
\end{align}
}
Note that $X_t$ follows the modified-ARCH process with unconditional CDF $\Gamma$. Now we can identify that $X_t$ corresponds to the latent AR process and extract the copula associated to the process by applying the PIT. To this end we define $W_t = \Gamma(X_t) \sim U(0,1)$, \SP{and $\varepsilon_t = \Phi(\eta_t)$ for all $t \in \mathbb{Z}$,} and we write
\begin{align}
W_t = \Gamma \left[ \sqrt{(\alpha_0 + \Gamma^{-1}(W_{t-1})^2)(\alpha_1 \Phi^{-1}(\varepsilon_t)^2 + \beta_1)} \right].
\end{align}
Following the procedure from the proof in~\ref{App:Proof:ARCH(1)}, we identify the above mapping as a conditional quantile function corresponding to a bivariate copula $C_{21}$, i.e.,
\begin{align}
C_{2|1}^{-1}(s|v) =& \Gamma \left[ \sqrt{(\alpha_0 + \Gamma^{-1}(v)^2)(\alpha_1 \Phi^{-1}(s)^2 + \beta_1)} \right],
\\
C_{2|1}(u|v) =& \Phi\left[ \sqrt{\frac{\Gamma^{-1}(u)^2}{\alpha_1\alpha_0 + \alpha_1\Gamma^{-1}(v)^2} - \frac{\beta_1}{\alpha_1}} \right]
\end{align}
Integration over $v$ yields the desired result. Now for the MAG-copula, we again consider the relation between $Y_t$ and $X_t$ and define $U_t = \Psi(Y_t)$. Then we can write
\begin{align}
U_t = \Psi \left[ \sqrt{\alpha_0 + \Gamma^{-1}(W_{t-1})^2} \Phi^{-1}(\varepsilon_t) \right].
\end{align}
Again, we identify the above mapping as a conditional quantile function corresponding to a copula $K_{21}$. That is
\SP{
\begin{align}
K_{2|1}^{-1}(s|v) =& \Psi \left[ \sqrt{\alpha_0 + \Gamma^{-1}(v)^2} \Phi^{-1}(s) \right] \quad \text{and}
\\
K_{2|1}(u|v) =& \Phi\left[ \frac{\Psi^{-1}(u)}{\sqrt{\alpha_0 + \Gamma^{-1}(v)^2}} \right].
\end{align}
}
Again integration over $v$ yields the desired result and the proof is complete.

\subsection{Proof of Proposition~\ref{prop:Distributional_MAG(1)}}
\label{App:Proof:MAG(1)}

\begin{proof}
We show properties step-by-step.
\begin{itemize}
\item[a)] The process $\{V_t\}$ given as \SP{$V_t = K_{2|1}^{-1}(\varepsilon_t|\varepsilon_{t-1})$} is $1$-dependent with $i.i.d.$ innovations and hence stationary and ergodic. See also \cite{joe2014dependence}, Sect. 3.14, and \cite{pappert2024moving}.
\item[b)] The derivation is given in \cite{pappert2024moving}. 
We still derive $F_{V_t,V_{t-1}}$ explicitly for the sake of completeness.
\begin{align}
F_{V_t,V_{t-1}}(x_1,x_2) =& \mathbb{P}(V_t \leq x_1, V_{t-1} \leq x_2)
\\
=& \mathbb{P}(K_{2|1}^{-1}(\varepsilon_t,\varepsilon_{t-1}) \leq x_1, K_{2|1}^{-1}(\varepsilon_{t-1},\varepsilon_{t-2}) \leq x_2)
\\
=& \int_0^1 \int_0^1 \mathbb{P}(K_{2|1}^{-1}(\varepsilon_t,w_{1}) \leq x_1, K_{2|1}^{-1}(w_{1},w_{2}) \leq x_2) \rmd w_1 \rmd w_2
\\
=& \int_0^1 \int_0^1 K_{2|1}(x_1,w_1) \mathds{1}\{ w_1 \leq K_{2|1}(x_2|w_2) \} \rmd w_1 \rmd w_2
\\
=& \int_0^1 K_{21}(x_1, K_{2|1}(x_2|w_2) \rmd w_2
\end{align}
\item[c)] We recall that $K$ is PQD, if $K(x_1,x_2) \geq x_1x_2, \forall x_1,x_2 \in \SP{(0,1)}$, cf.~\cite{joe2014dependence}, Sect.~2.8.1. Then we see that
\begin{align*}
F_{V_t,V_{t-1}}(x_1,x_2) =& \int_0^1 K_{12}(K_{2|1}(x_1|s), x_1) \rmd s 
\\
\geq &  \int_0^1 K_{2|1}(x_2|s) x_1 \rmd s
%\\
%=& x_1x_2.
= x_1x_2.
\end{align*}
In the first inequality we used the PQD of $K$ and in the second equality we solved the integral. For $K$ being NQD, we can show that $F_{V_t,V_{t-1}}$ is NQD in the same way.
\item[d)] We recall that $K_{2|1}$ is SI %, in its first argument, 
if $K_{2|1}(x_2|x_1)$ is decreasing in $x_1$. Alternatively, consider the random vector $(Z_1,Z_2)$ which we define in such a way that $(Z_1,Z_2) \sim K$. Then $K$ being SI in $Z_1$ is equivalent to $Z_2$ being SI in $Z_1$. 
%We first realize that 
If $K_{2|1}(x_2|x_1)$ is decreasing in $x_1$, then $K^{-1}_{2|1}(y|x_1)$ is increasing in $x_1$. Then  $V_t = K_{2|1}^{-1}(\varepsilon_t, \varepsilon_{t-1})$ is SI in $\varepsilon_{t-1}$. Additionally, $\varepsilon_{t-1} = K_{2|1}(V_{t-1}|\varepsilon_{t-2})$ is increasing in $V_{t-1}$. In combination, $V_t$ is increasing in $V_{t-1}$.
\item[e)] Let $R$ be the infinite dimensional correlation matrix for $\{U_t\}_{t\in\mathbb{Z}}$. It holds that $R_{ii} = 1, \, \forall i \in \mathbb{Z}$ and $R_{i,i\pm 1} = \rho$. All elements that are not on the main and secondary diagonal are zero since the process is one-dependent. Hence $R$ is a tridiagonal matrix and we can use the result from \cite{tao2025tridiagonal} that establishes that $|\rho|\leq\frac{1}{2}$ has to hold in order for $R$ to be a correlation matrix.
\item[f)] We use the result from \cite{embrechts2016bernoulli} which establishes that $|\lambda_u|, |\lambda_l| \leq \frac{1}{2}$ for general $1$-dependent processes.
\item[g)] For $Y_t = \{\Phi^{-1}(V_t)\}$ it holds that $\{Y_t\}_{t\in\mathbb{Z}}$ follows the updating equation (compare Proposition~\ref{prop:GaussianARMA}),
\begin{align*}
Y_t = \eta_t + \frac{\beta_1}{(1 - \beta_1)^{1/2}} \eta_{t-1}, \quad \eta_t \sim N(0,(1-\beta_1^2))
\end{align*}
For $\tilde{Y}_t = \{\Phi^{-1}(\tilde{V}_t)\}$ is holds that
\begin{align*}
\tilde{Y}_t = \tilde{\eta}_t + \frac{(1-\beta_1^2)^{1/2}}{\beta_1} \tilde{\eta}_{t-1}, \quad \tilde{\eta}_t \sim N(0,\beta_1^2).
\end{align*}
It is well known that $\{Y_t\}$ and $\{\tilde{Y}_t\}$ are two representations of the same MA$(1)$ process. The proposition follows immediately.
\end{itemize}
\end{proof}

\subsection{Proof of Proposition~\ref{prop:CICARMA(1,1)}}
\label{App:Proof:CICARMA(1,1)}

\begin{proof}
We go step-by-step.
\begin{itemize}
\item[a)] The time series $\{W_t\}$ follows an copula AR$(1)$ process, \SP{with $C_{21}$ being the copula of $(W_t,W_{t-1})$ for all $t$.} Assume that $C_{21}$ is chosen such that the process is stationary and ergodic (this is true under mild conditions). Then it follows that $\{U_t\}$ is also stationary and ergodic since it is calculated by aggregating the independent random variables $\varepsilon_t$ and $W_{t-1}$, \SP{i.e., $U_t = K_{2|1}^{-1}(\varepsilon_t|W_{t-1})$, implying that $K_{21}$ is the copula of $(U_t,W_{t-1})$.}
\item[b)] We consider the joint distribution of $(U_t, U_{t-1})$,
\begin{align*}
F_{U_t,U_{t-1}}(x_1, x_2) =& \mathbb{P}(U_t \leq x_1, U_{t-1} \leq x_2)
\\
=& \mathbb{P}(h(\varepsilon_t\, ,\, W_{t-1}) \leq x_1, h(\varepsilon_{t-1}, W_{t-2}) \leq x_2)
\\
=& \mathbb{P}(\varepsilon_t \leq K_{2|1}(x_1|W_{t-1}) \, ,\, \varepsilon_{t-1} \leq K_{2|1}(x_2|W_{t-2}))
\\
=&
\mathbb{P}(\varepsilon_t \leq K_{2|1}(x_1|W_{t-1}) \, ,\, W_{t-1} \leq C_{2|1}^{-1}(K_{2|1}(x_2|W_{t-2})| W_{t-2}))
\\
=& \int_0^1 \int_0^1 c_{21}(w_1,w_2) \mathbb{P}(\varepsilon_t \leq K_{2|1}(x_1|w_{1}) \, ,\, w_{1} \leq C_{2|1}^{-1}(K_{2|1}(x_2|w_{2})| w_{2})) \, \rmd w_1  \rmd w_2
\\
=& \int_0^1 \int_0^1 c_{21}(w_1,w_2) \, K_{2|1}(x_1|w_{1}) \, 1_{\{ w_{1} \leq C_{2|1}^{-1}(K_{2|1}(x_2|w_{2})| w_{2}) \}} \, \rmd w_1  \rmd w_2
\\
=& \int_0^1 \int_0^{C_{2|1}^{-1}(K_{2|1}(x_2|w_{2})| w_{2})} c_{21}(w_1,w_2) \, K_{2|1}(x_1|w_{1})  \rmd w_1  \rmd w_2,
\end{align*}
where we used that the joint density of $(W_t, W_{t-1})$ is given by the copula density $c_{21}$.
\SP{
\item[c)] We can decompose the joint density of $(U_t, W_{t-1}, W_t, U_{t+1})$ as follows
\begin{align*}
    f(x_1, x_2, x_3, x_4) =& c_{U_t, W_{t-1}}(x_1, x_2) c_{W_{t-1}, W_{t}}(x_2, x_3) c_{W_t, U_{t+1}}(x_3, x_4) 
    \\
    \times & c_{U_t, W_t;W_{t-1}}\left(F_{U_t|W_{t-1}}(x_1|x_2), F_{W_t|W_{t-1}}(x_3|x_2)\right)
    \\
    \times & c_{W_{t-1}, U_{t+1};W_{t}}\left(F_{W_{t-1}|W_t}(x_2|x_3), F_{U_{t+1}|W_{t}}(x_4|x_3)\right)
    \\
    \times & c_{U_t U_{t+1} ; W_{t-1}, W_{t}}\left(F_{U_t|W_{t-1}, W_t}(x_1 | x_2, x_3), F_{U_{t+1}|W_{t-1}, W_t}(x_4| x_2, x_3)\right).
    \end{align*}
By definition the copula densities $c_{U_t, W_{t-1}}$ and $c_{W_t, U_{t+1}}$ are given by the copula density associated to $K_{12}$, namely $k_{12}$. The copula density $c_{W_{t-1}, W_{t}}$ is given by $c_{12}$. 
%The copulas $C_{U_t, W_t;W_{t-1}}$, $C_{W_{t-1}, U_{t+1};W_{t}}$ and $ C_{U_t U_{t+1} ; W_{t-1}, W_{t}}$ can be derived explicitly. 
It can be shown that $C_{U_t, W_t;W_{t-1}} = C^{+}$ (comonotonicity), $C_{W_{t-1}, U_{t+1};W_{t}} = C^{\perp}$ (independence) and $C_{U_t U_{t+1} ; W_{t-1}, W_{t}}$ is arbitrary because $U_t$ is completely determined by $(W_t = w_t, W_{t-1} = w_{t-1})$. So $c_{U_t U_{t+1} ; W_{t-1}, W_{t}}$ can be taken as 1.
\item[d)] The result immediately follows from Theorem 4.1 by \cite{joe2010tail}.
}
\end{itemize}
\end{proof}

\subsection{Proof of Proposition~\ref{prop:Joint_Distributions_higher_orders}}
\label{App:Proof:Joint_Distribution_higher_orderss}

\begin{proof}
For $p=1, q=2$ we have \SP{that $C_{21}$ is the copula of $(W_t, W_{t-1})$ and $K_{321}$ the copula of $(U_t, \varepsilon_{t-1}, W_{t-2})$. The joint distribution of consecutive observations can be calculated as follows.} 
\begin{align*}
F_{U_t,U_{t-1}}(x_1, x_2) =& \mathbb{P}(U_t \leq x_1, U_{t-1} \leq x_2)
\\
=& \mathbb{P}(K_{3|21}^{-1}(\varepsilon_t|\varepsilon_{t-1},W_{t-2}) \leq x_1, K_{3|21}^{-1}(\varepsilon_{t-1}|\varepsilon_{t-2},W_{t-3}) \leq x_2)
\\
=& \mathbb{P}\bigl(\varepsilon_t \leq K_{3|21}(x_1|\varepsilon_{t-1}, W_{t-2}), \varepsilon_{t-1} \leq K_{3|21}(x_2|\varepsilon_{t-2}, W_{t-3})\bigr).
\end{align*}
Now we condition on suitable random variables. It is important that the joint distribution of the random variables which the expression is conditioned on is known. Here we condition on $(\varepsilon_{t-1},W_{t-2},W_{t-3})$. The joint density of this random vector is $f_{\varepsilon_{t-1},W_{t-2},W_{t-3}}(u_1, u_2, u_3) =c_{21}(u_2, u_3)$ by construction, because $\varepsilon_{t-1}$ is independent of $(W_{t-2}, W_{t-3})$. Hence we can continue the calculation as follows.
\begin{align*}
F_{U_t,U_{t-1}}(x_1, x_2) =& \SP{\int_0^1 \int_0^1 \int_0^1} \mathbb{P}(\varepsilon_t \leq K_{3|21}(x_1|e_1, w_2), e_1 \leq K_{3|21}(x_2|e_2, w_3)) c_{21}(w_2, w_3)  \rmd e_1 \rmd w_2 \rmd w_3
\\
=& \SP{\int_0^1 \int_0^1 \int_0^1} K_{3|21}(x_1|e_1, w_2) \times \mathds{1}\{ e_1 \leq K_{3|21}(x_2|e_2, w_3) \} c_{21}(w_2, w_3) \SP{\rmd e_1 \rmd w_2 \rmd w_3.}
\end{align*}
Now for $p=2, q=1$ we have \SP{that $C_{321}$ is the copula of $(W_t, W_{t-1}, W_{t-2})$ and $K_{21}$ the copula of $(U_t, W_{t-1})$. The joint distribution of consecutive observations can be calculated as follows.} 
\begin{align*}
F_{U_t,U_{t-1}}(x_1, x_2) =& \mathbb{P}(U_t \leq x_1, U_{t-1} \leq x_2)
\\
=& \mathbb{P}(K_{2|1}^{-1}(\varepsilon_t| W_{t-1}) \leq x_1, K_{2|1}^{-1}(\varepsilon_{t-1}| W_{t-2}) \leq x_2)
\\
=& \mathbb{P}(\varepsilon_t \leq K_{2|1}(x_1|W_{t-1}), \varepsilon_{t-1} \leq K_{2|1}(x_2|W_{t-2}))
\\
=& \mathbb{P}(\varepsilon_t \leq K_{2|1}(x_1|W_{t-1}), W_{t-1} \leq C_{3|21}^{-1}(K_{2|1}(x_2|W_{t-2}))|W_{t-2}, W_{t-3}))
\end{align*}
Now we condition on the random vector $(W_t, W_{t-1}, W_{t-2})$, which has joint density given by $c_{321}$.
\begin{align*}
F_{U_t,U_{t-1}}(x_1,x_2) =& \int_0^1 \int_0^1 \int_0^1 \Big[ K_{2|1}(x_1|w_1) c_{321}(w_1,w_2,w_3) \mathds{1}\{ w_1 \leq C_{3|21}^{-1}(K_{2|1}(x_2|w_2) | w_2, w_3) \}
\\
&\times c_{321}(w_1, w_2, w_3) \Big] \rmd w_1 \rmd w_2 \rmd w_3 .
\end{align*}
For $p=q=2$ we have \SP{that $C_{321}$ is the copula of $(W_t, W_{t-1}, W_{t-2})$ and $K_{321}$ the copula of $(U_t, \varepsilon_{t-1}, W_{t-2})$. The joint distribution of consecutive observations can be calculated as follows.} 
\begin{align*}
F_{U_t,U_{t-1}}(x_1, x_2) =& \mathbb{P}(U_t \leq x_1, U_{t-1} \leq x_2)
\\
=& \mathbb{P}(\varepsilon_t \leq K_{3|21}(x_1|\varepsilon_{t-1}, W_{t-2}), \varepsilon_{t-1} \leq K_{3|21}(x_2|\varepsilon_{t-2}, W_{t-3}))
\end{align*}
We substitute $\varepsilon_{t-2} = C_{3|21}(W_{t-2}|W_{t-3}, W_{t-4})$ and condition on $(\varepsilon_{t-1}, W_{t-2}, W_{t-3}, W_{t-4})$ with joint density $c_{321}$.
\begin{align*}
F_{U_t,U_{t-1}}(x_1, x_2) =& \int_0^1 \int_0^1 \int_0^1 \int_0^1 K_{3|21}(x_1|e_1, w_2) c_{321}(w_2, w_3, w_4) \mathds{1}\{ e_1 \leq K_{3|21}(x_2|e_2, w_3) \} \rmd w_2 \rmd w_3 \rmd w_4 \rmd e_1.
\end{align*}
\end{proof}

\end{document}